\documentclass[12pt]{article}

\def\UseBibLatex{1}

\makeatletter
\def\input@path{{styles/}}
\makeatother

\providecommand{\BibLatexMode}[1]{}
\providecommand{\BibTexMode}[1]{}

\renewcommand{\BibLatexMode}[1]{#1}
\renewcommand{\BibTexMode}[1]{}

\ifx\UseBibLatex\undefined%
  \renewcommand{\BibLatexMode}[1]{}
  \renewcommand{\BibTexMode}[1]{#1}
\fi

\BibLatexMode{%
   \usepackage[bibencoding=utf8,style=alphabetic,backend=biber]{biblatex}%
   \usepackage{sariel_biblatex}%
}

\usepackage[cm]{fullpage}%
\usepackage{amsmath}%
\usepackage{amssymb}%
\usepackage[table]{xcolor}%

\usepackage[amsmath,thmmarks]{ntheorem}%

\usepackage{titlesec}%
\usepackage{xcolor}%
\usepackage{mleftright}%
\usepackage{xspace}%
\usepackage{graphicx}
\usepackage{hyperref}%
\usepackage[inline]{enumitem}
\usepackage{hyperref}%
\usepackage[ocgcolorlinks]{ocgx2}
\usepackage{euscript}
\usepackage{wrapfig}  %
\usepackage{caption}

\hypersetup{%
      unicode,
      breaklinks,%
      colorlinks=true,%
      urlcolor=[rgb]{0.25,0.0,0.0},%
      linkcolor=[rgb]{0.5,0.0,0.0},%
      citecolor=[rgb]{0,0.2,0.445},%
      filecolor=[rgb]{0,0,0.4},
      anchorcolor=[rgb]={0.0,0.1,0.2}%
}

\titlelabel{\thetitle. }%

\theoremseparator{.}%

\theoremstyle{plain}%
\newtheorem{theorem}{Theorem}[section]

\newtheorem{lemma}[theorem]{Lemma}

\newtheorem{claim}[theorem]{Claim}%

\theoremstyle{plain}%
\theoremheaderfont{\sf} \theorembodyfont{\upshape}%
\newtheorem*{remark:unnumbered}[theorem]{Remark}%
\newtheorem{remark}[theorem]{Remark}%

\newtheorem{defn}[theorem]{Definition}

\theoremheaderfont{\em}%
\theorembodyfont{\upshape}%
\theoremstyle{nonumberplain}%
\theoremseparator{}%
\theoremsymbol{\myqedsymbol}%
\newtheorem{proof}{Proof:}%

\providecommand{\emphind}[1]{}%
\renewcommand{\emphind}[1]{\emph{#1}\index{#1}}

\definecolor{blue25emph}{rgb}{0, 0, 11}

\providecommand{\emphic}[2]{}
\renewcommand{\emphic}[2]{\textcolor{blue25emph}{%
      \textbf{\emph{#1}}}\index{#2}}

\providecommand{\emphi}[1]{}%
\renewcommand{\emphi}[1]{\emphic{#1}{#1}}

\definecolor{almostblack}{rgb}{0, 0, 0.3}

\providecommand{\emphw}[1]{}%
\renewcommand{\emphw}[1]{{\textcolor{almostblack}{\emph{#1}}}}%

\providecommand{\emphOnly}[1]{}%
\renewcommand{\emphOnly}[1]{\emph{\textcolor{blue25emph}{\textbf{#1}}}}

\newcommand{\atgen}{\symbol{'100}}
\newcommand{\myqedsymbol}{\rule{2mm}{2mm}}
\newcommand{\SarielThanks}[1]{%
   \thanks{%
      School of Computing and Data Science; %
      University of Illinois; %
      201 N. Goodwin Avenue; %
      Urbana, IL, 61801, USA; %
      \href{mailto:spam@illinois.edu}{sariel@illinois.edu}; %
      \url{http://sarielhp.org/}.%
   #1%
   }%
}

\newcommand{\HLink}[2]{\hyperref[#2]{#1~\ref*{#2}}}
\newcommand{\HLinkSuffix}[3]{\hyperref[#2]{#1\ref*{#2}{#3}}}

\newcommand{\figlab}[1]{\label{fig:#1}}
\newcommand{\figref}[1]{\HLink{Figure}{fig:#1}}

\newcommand{\thmlab}[1]{{\label{theo:#1}}}
\newcommand{\thmref}[1]{\HLink{Theorem}{theo:#1}}

\newcommand{\seclab}[1]{\label{sec:#1}}
\newcommand{\secref}[1]{\HLink{Section}{sec:#1}}

\newcommand{\apndlab}[1]{\label{apnd:#1}}
\newcommand{\apndref}[1]{\HLink{Appendix}{apnd:#1}}

\newcommand{\remlab}[1]{\label{rem:#1}}
\newcommand{\remref}[1]{\HLink{Remark}{rem:#1}}%

\newcommand{\clmlab}[1]{\label{claim:#1}}
\newcommand{\clmref}[1]{\HLink{Claim}{claim:#1}}

\providecommand{\deflab}[1]{\label{def:#1}}
\renewcommand{\deflab}[1]{\label{def:#1}}
\newcommand{\defref}[1]{\HLink{Definition}{def:#1}}

\newcommand{\lemlab}[1]{\label{lemma:#1}}
\newcommand{\lemref}[1]{\HLink{Lemma}{lemma:#1}}%

\providecommand{\eqlab}[1]{}%
\renewcommand{\eqlab}[1]{\label{equation:#1}}
\newcommand{\Eqref}[1]{\HLinkSuffix{Eq.~(}{equation:#1}{)}}

\providecommand{\remove}[1]{}%

\newcommand{\pth}[1]{\mleft(#1\mright)}%

\newcommand{\ProbC}{{\mathbb{P}}}
\newcommand{\ExC}{{\mathbb{E}}}

\newcommand{\Prob}[1]{\ProbC\mleft[ #1 \mright]}
\newcommand{\Ex}[1]{\ExC\mleft[ #1 \mright]}

\newcommand{\ceil}[1]{\mleft\lceil {#1} \mright\rceil}
\newcommand{\floor}[1]{\mleft\lfloor {#1} \mright\rfloor}

\newcommand{\brc}[1]{\left\{ {#1} \right\}}

\newcommand{\cardin}[1]{\left\lvert {#1} \right\rvert}%

\renewcommand{\th}{th\xspace}
\newcommand{\ds}{\displaystyle}%

\renewcommand{\Re}{\mathbb{R}}%

\newlist{compactenumA}{enumerate}{5}%
\setlist[compactenumA]{topsep=0pt,itemsep=-1ex,partopsep=1ex,parsep=1ex,%
   label=(\Alph*)}%

\newlist{compactenuma}{enumerate}{5}%
\setlist[compactenuma]{topsep=0pt,itemsep=-1ex,partopsep=1ex,parsep=1ex,%
   label=(\alph*)}%

\newlist{compactenumI}{enumerate}{5}%
\setlist[compactenumI]{topsep=0pt,itemsep=-1ex,partopsep=1ex,parsep=1ex,%
   label=(\Roman*)}%

\newlist{compactenumi}{enumerate}{5}%
\setlist[compactenumi]{topsep=0pt,itemsep=-1ex,partopsep=1ex,parsep=1ex,%
   label=(\roman*)}%

\newlist{compactitem}{itemize}{5}%
\setlist[compactitem]{topsep=0pt,itemsep=-1ex,partopsep=1ex,parsep=1ex,%
   label=\ensuremath{\bullet}}%

\numberwithin{figure}{section}%
\numberwithin{table}{section}%
\numberwithin{equation}{section}%

\providecommand{\parent}{\overline{\mathrm{p}}}
\providecommand{\parentX}[1]{\parent\pth[]{#1}}

\newcommand{\depthX}[1]{\mathrm{depth}\pth{#1}}

\newcommand{\Spread}{\Phi}
\newcommand{\SpreadX}[1]{\Phi\pth{#1}}
\providecommand{\TPDF}[2]{\texorpdfstring{#1}{#2}}

\newcommand{\Cell}{\Box}

\newcommand{\ropt}{r_{\mathrm{opt}}}

\renewcommand{\th}{t{}h\xspace}

\newcommand{\WSPDRep}{\mathcal{W}}
\newcommand{\SW}{\EuScript{S}}

\providecommand{\ds}{\displaystyle}

\newcommand{\WSPD}{\textsf{WSPD}\xspace}
\newcommand{\SSPD}{\textsf{SSPD}\xspace}
\newcommand{\BAR}{\textsf{BAR}\xspace}

\newcommand{\SSPDs}{\textsf{SSPD}s\xspace}

\newcommand{\pnt}{\mathsf{p}}
\newcommand{\pntA}{\mathsf{q}}
\newcommand{\pntB}{\mathsf{r}}
\newcommand{\pntC}{\mathsf{z}}

\newcommand{\PntSet}{\mathsf{P}}
\newcommand{\PntSetA}{\mathsf{Q}}
\newcommand{\PntSetB}{{\mathsf{R}}}
\newcommand{\PntSetC}{{\mathsf{S}}}

\providecommand{\pbrcx}[1]{\left[ {#1} \right]}
\providecommand{\Ex}[1]{\mathop{\mathbf{E}}\!\pbrcx{#1}}
\providecommand{\Prob}[1]{\mathop{\mathbf{P{}r}}\!\pbrcx{#1}}

\newcommand{\Pin}{\PntSet_{\mathrm{in}}}
\newcommand{\Pout}{\PntSet_{\mathrm{out}}}
\newcommand{\Pouter}{\PntSet_{\mathrm{outer}}}

\newcommand{\Pring}{\PntSet_{\mathrm{ring}}}
\newcommand{\ring}{\mathcal{R}}

\providecommand{\MakeBig}{\rule[-.2cm]{0cm}{0.4cm}}
\providecommand{\MakeSBig}{\rule[0.0cm]{0.0cm}{0.35cm}} %

\providecommand{\pbrc}[1]{\!\!\left[ {#1} \MakeBig \right]}

\providecommand{\cardin}[1]{\left| {#1} \right|}
\providecommand{\ceil}[1]{\left\lceil {#1} \right\rceil}
\providecommand{\pth}[2][\!]{#1\left({#2}\right)}

\providecommand{\si}[1]{#1}

\newcommand{\DiamX}[1]{\mathrm{\si{diam}}\pth{#1}}

\newcommand{\MTR}{\mathcal{M}}
\newcommand{\BallC}{\mathsf{b}}

\newcommand{\BallX}[2]{\mathsf{b}\pth{ #1, #2  }}
\newcommand{\RingX}[3]{\mathrm{ring}\pth{ #1, #2, #3 \MakeSBig}}

\newcommand{\Graph}{\ensuremath{\mathcal{G}}}

\newcommand{\DistChar}{\mathsf{d}}
\newcommand{\DistSetX}[2]{\DistChar\pth{#1,#2}}

\providecommand{\dist}[1]{\left\| {#1} \right\|}
\newcommand{\distX}[2]{\left\| {#1 #2} \right\|}
\newcommand{\distG}[2]{\DistChar_\Graph\pth{ #1, #2}}
\definecolor{blue25}{rgb}{0,0,0.25} \providecommand{\emphic}[2]{%
   \textcolor{blue25}{%
      \textbf{\emph{#1}}}%
   \index{#2}}

\providecommand{\emphi}[1]{\emphic{#1}{#1}}
\providecommand{\eps}{{\varepsilon}}%
\providecommand{\brc}[1]{\left\{ {#1} \right\}}
\providecommand{\sep}[1]{\,\left|\, {#1} \MakeBig\right.}

\renewcommand{\Re}{\mathbb{R}}%

\newcommand{\WeightX}[1]{\omega \pth{#1}}
\newcommand{\QTree}{\EuScript{T}}

\newcommand{\PSetX}[1]{\mathcal{S}\pth{#1}}

\newcommand{\cset}{\mathcal{C}}

\newcommand{\cone}{\sigma}

\newcommand{\hub}{\mathrm{hub}}

\providecommand{\etal}{\textit{e{t}~{a}l.}\xspace}

\newcommand{\coneAngle}{\psi}

\newcommand{\tspanner}{$t$-spanner\xspace}

\newcommand{\Furer}{F{\"u}r{e}r\xspace}
\newcommand{\Err}{\mathcal{E}}

\newcommand{\SetA}{\mathcal{X}}
\newcommand{\SetB}{\mathcal{Y}}
\newcommand{\SetC}{\mathcal{Z}}

\newcommand{\Family}{\EuScript{F}}

\newcommand{\regionAB}{R_{\pntA,\pntB}}
\newcommand{\regionRest}{R_{\mathrm{rest}}}

\newlength{\savedparindent}

\newcommand{\myParagraph}[1]{\paragraph{#1}}
\newcommand{\Edges}{\mathsf{E}}

\newcommand{\CX}{\EuScript{X}}%
\newcommand{\CY}{\EuScript{Y}}%
\newcommand{\Thickness}{\varpi}%

\newcommand{\src}{\mathsf{s}}
\newcommand{\target}{\mathsf{t}}

\BibLatexMode{%
   \bibliography{sspd}
}

\begin{document}

\title{New Constructions of \SSPD{}s and their Applications%
   \thanks{A preliminary version of this paper appeared in SoCG 2010
      \cite{ah-ncsa-10}.  The latest version of this paper is
      available online \cite{ah-ncsa-10}.}}

\author{Mohammad A. Abam%
   \thanks{Department of Computer Engineering, Sharif University of Technology, Tehran, Iran.
   email:{\tt abam\atgen{}sharif.edu}.}%
   \and %
   Sariel Har-Peled%
   \SarielThanks{Work on this paper was partially supported by a NSF
      AF award CCF-0915984.}  }

\date{March 29, 2018}

\maketitle

\begin{abstract}
    We present a new optimal construction of a semi-separated pair
    decomposition (i.e., \SSPD) for a set of $n$ points in $\Re^d$. In
    the new construction each point participates in a few pairs, and
    it extends easily to spaces with low doubling dimension. This is
    the first optimal construction with these properties.

    As an application of the new construction, for a fixed $t>1$, we
    present a new construction of a \tspanner with $O(n)$ edges and
    maximum degree $O(\log^2 n)$ that has a separator of size
    $O\pth{n^{1-1/d}}$.
\end{abstract}

\section{Introduction}

For a point-set $\PntSet$, a pair decomposition of $\PntSet$ is a set
$\WSPDRep$ of pairs of subsets of $\PntSet$, such that for every pair
of points of $\pnt, \pntA \in \PntSet$ there exists a pair
$\brc{\SetA, \SetB} \in \WSPDRep$ such that $\pnt \in \SetA$ and
$\pntA \in \SetB$---see \secref{preliminaries} for the formal
definition.  A well-separated pair decomposition (\WSPD) of $\PntSet$
is a pair decomposition of $\PntSet$ such that for every pair
$\brc{\SetA, \SetB}$, the distance between $\SetA$ and $\SetB$ is
large when compared to the {maximum} diameter of the two point
sets. The notion of \WSPD was developed by Callahan and
Kosaraju~\cite{ck-dmpsa-95}, and it provides a compact representation
of the quadratic pairwise distances of the point-set $\PntSet$, since
there is a \WSPD with a linear number of pairs.

The total weight of a pair decomposition $\WSPDRep$ is the total size
of the sets involved; that is $ \WeightX{\WSPDRep} = \sum_{\brc{\SetA,
      \SetB} \in \WSPDRep} \pth{ \cardin{\SetA} + \cardin{\SetB}}$.
Naturally, a \WSPD with near linear weight is easier to manipulate and
can be used in applications where the total weight effects the overall
performance.  Unfortunately, it is easy to see that, in the worst
case, the total weight of any \WSPD is $\Omega\pth{n^2}$.  Callahan
and Kosaraju~\cite{ck-dmpsa-95} overcame this issue by generating an
implicit representation of the \WSPD using a tree. Indeed, they build
a tree $\QTree$ (usually a compressed quadtree, or some other
variant) that stores the points of the given point set in the
leaves. Pairs are reported as $\brc{u,v}$ where $u$ and $v$ are nodes
in the tree such that their respective pair $\PSetX{u} \otimes
\PSetX{v}$ is well separated, where $\PSetX{u}$ denotes the set of
points stored at the subtree of $u$.

\myParagraph{\SSPD{}s.}  To overcome this obesity problem, a weaker notion of {semi-separated pair decomposition} (\SSPD) had been suggested by Varadarajan \cite{v-dcamc-98}.  Here, an \SSPD $\SW$ of a point set has the property that for each pair $\brc{\SetA, \SetB} \in \SW$ the distance between $\SetA$ and $\SetB$ is large when compared to the \emph{minimum} diameter of the two point sets (in the \WSPD case this was the \emph{maximum} of the diameters).  See \figref{sspd_ex} for an example of a semi-separated pair that is not well-separated.

\begin{wrapfigure}{r}{0.25\linewidth}
    \centering
    \includegraphics{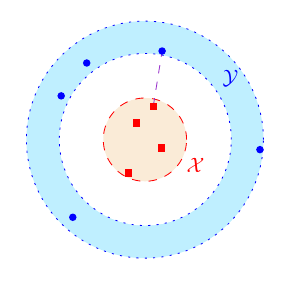}
    \caption{}
    \figlab{sspd_ex}
\end{wrapfigure}

By weakening the separation, one can get an \SSPD with near-linear
total weight. Specifically, Varadarajan~\cite{v-dcamc-98} showed how
to compute an \SSPD of weight $O\pth{n \log^4 n}$ for a set of $n$
points in the plane, in $O(n\log^5 n)$ time (for a constant separation
factor).  He used the \SSPD for speeding up his algorithm for the
min-cost perfect-matching in the plane. Recently, Abam \etal
\cite{adfg-rftgs-09} presented an algorithm which improves the
construction time to $O\pth{\eps^{-2} n\log n}$ and the weight to
$O\pth{ \eps^{-2} n\log n}$, where $1/\eps$ is the separation required
between pairs. It is known that any pair decomposition has weight
$\Omega\pth{ n\log n }$ \cite{bs-ss-07}, implying that the result of
Abam \etal \cite{adfg-rftgs-09} is optimal.  This construction was
generalized to $\Re^d$ with the construction time being $O(\eps^{-d}
n\log n)$ and the total weight of the \SSPD being $O(\eps^{-d} n\log
n)$, see \cite{adfgs-gswps-11}.

\paragraph{Spanners.}

More recently \SSPDs were used in constructing certain geometric
spanners \cite{adfg-rftgs-09, adfgs-gswps-11, acfs-psspd-09}.  Let
$\Graph=(\PntSet, \Edges)$ be a geometric graph on a set $\PntSet$ of
$n$ points in $\Re^d$.  That is, $\Graph$ is an edge-weighted graph
where the weight of an edge $\pnt\pntA \in \Edges$ is the Euclidean
distance between $\pnt$ and $\pntA$. The distance in $\Graph$ between
two points $\pnt$ and $\pntA$, denoted by $\distG{\pnt}{\pntA}$, is
the length of the shortest (that is, minimum-weight) path between
$\pnt$ and $\pntA$ in $\Graph$. A graph $\Graph$ is a (geometric)
\emphi{\tspanner{}}, for some $t\geq 1$, if for any two points $\pnt,
\pntA \in \PntSet$ we have $\distG{\pnt}{\pntA} \leq t\cdot
\distX{\pnt}{\pntA}$. Note that the concept of geometric spanners can
be easily extended to any metric space.  Geometric spanners have
received considerable attention in the past few years---see
\cite{ns-gsn-07} and references therein.  Obviously, the complete
graph is a \tspanner, but a preferable spanner would provide short
paths between its nodes, while having few edges. Other properties
considered include
\begin{compactenumi}
    \item low total length of edges,
    \item low diameter,
    \item low maximum degree, and
    \item having small separators.
\end{compactenumi}

\paragraph{Separators.}

A graph $\Graph$ is said to have a \emphi{$k$-separator} if its
vertices can be decomposed into three sets $\SetA$, $\SetB$ and
$\SetC$ such that both $\cardin{\SetA}$ and $\cardin{\SetB}$ are
$\Omega(n)$ and $\cardin{\SetC} \leq k$ (i.e., the set $\SetC$ is the
\emph{separator}). Furthermore, there is \emph{no} edge in $\Graph$
connecting a vertex in $\SetA$ to a vertex in $\SetB$.  Graphs with
good separators are the best candidates to applying the divide and
conquer approach---see \cite{lt-apst-80, mttv-sspnng-97, sw-gsta-98}
and references therein for more results and applications. Lipton and
Tarjan \cite{lt-apst-80} showed that any planar graph has an $O\pth{
   \!  \sqrt{n} \, \MakeSBig }$-separator.  The Delaunay triangulation
of a planar point-set $\PntSet$ is an $O(1)$-spanner
\cite{kg-cgwaceg-92}. Furthermore, it is a planar graph and is thus an
$O(1)$-spanner with an $O(\sqrt{n})$-separator. \Furer and
Kasiviswanathan \cite{fk-sgig-07} recently presented a \tspanner for a
ball graph which is an intersection graph of a set of $n$ ball in
$\Re^d$ with arbitrary radii which has an
$O\pth{n^{1-1/d}}$-separator. Since complete Euclidean graphs are a
special case of unit ball graphs, their results yields a new
construction of \tspanner for geometric graphs with small separators.

\paragraph{Metric spaces with low doubling dimension.}
Recently the notion of doubling dimension \cite{a-pldr-83,
   gkl-bgfld-03, h-lams-01} which is a generalization of the Euclidean
dimension, has received considerable attention.  The \emph{doubling
   constant} of a metric space $\MTR$ is the maximum, over all balls
$\BallC$ in the metric space $\MTR$, of the minimum number of balls
needed to cover $\BallC$, using balls with half the radius of
$\BallC$. The logarithm of the doubling constant is the
\emphi{doubling dimension} of the space---note that $\Re^d$ has
$\Theta(d)$ doubling dimension.  Constructions of \WSPD, spanners, and
a data-structure for approximate nearest neighbor for metric spaces
with fixed doubling dimension were provided by Har-Peled and Mendel
\cite{hm-fcnld-06}.

\paragraph{Limitations of known constructions of \SSPDs.}
The optimal construction of \SSPDs of \cite{adfg-rftgs-09} uses
\BAR-trees \cite{dgk-bartc-01}, and as such it is not applicable to
metric spaces with constant doubling dimension. Moreover, in all
previously known optimal constructions of \SSPDs a point might appear
in many (i.e., linear number of) pairs, see \apndref{lower:bound} for
an example demonstrating this. Note, that in Varadarajan's original
construction \cite{v-dcamc-98} a point might appear in a $O( \log^4
n)$ pairs (for the planar case).  In particular, in all previous
constructions of spanners with small separators the degree of a point
might be $\Omega(n)$.

\paragraph{Our results.}
Interestingly, building an \SSPD for a point set with polynomially
bounded spread is relatively easy, as we point out in
\secref{preliminaries} (see \remref{low:spread:easy}). However,
extending this construction to the unbounded spread is more
challenging. Intuitively, this is because the standard \WSPD
construction is local and greedy in nature, and does not take the
global structure of the point set into account (in particular, it
ignore weight issues all together), which is necessary when handling
unbounded spread.

Building upon this bounded spread construction, we present two new
constructions of \SSPDs for a point set in $\Re^d$.  The first
construction guarantees that each point appears in $O\pth{\log^2 n}$
pairs and the resulting \SSPD has weight $O\pth{n\log^2 n}$. This
construction is (arguably) simpler than previous constructions.

The second construction uses a randomized partition scheme, which is
of independent interest, and results in an optimal \SSPD where each
point appears in $O(\log n)$ pairs, with high probability. This is the
first construction to guarantee that no point participates in too many
pairs. The two new constructions of \SSPDs work also in finite metric
spaces with low doubling dimension. This enables us to extend, in a
plug and play fashion, several results from Euclidean space to spaces
with low doubling dimension---see \secref{immediate:applications} for
details.

We also present an algorithm for constructing a spanner in Euclidean
space from any \SSPD. Since this is a simple generalization of the
algorithm given in~\cite{adfg-rftgs-09}, we delegate the description
of this algorithm to \apndref{SSPD:to:spanners}. The new proof showing
that the output graph is \tspanner is independent of the \SSPD
construction, unlike the previous proof of Abam \etal
\cite{adfg-rftgs-09}.

Using the new constructions, for a fixed $t$, we present a new
construction of a \tspanner with $O(n)$ edges and maximum degree
$O\pth{\log^2 n}$. The construction of this spanner takes
$O\pth{n\log^2 n}$ time, and it contains an $O\pth{ n^{1-1/d}
}$-separators. Note that, in the worst case, the separator can not be
much smaller.  The previous construction of \Furer and Kasiviswanathan
\cite{fk-sgig-07} is slower, and the maximum degree in the resulting
spanner is unbounded---see \thmref{separator} for details.

\bigskip

The paper is organized as follows: In \secref{preliminaries}, we
formally define some of the concepts we need and prove some basic
lemmas.  \secref{simple} presents a simple construction of \SSPDs.  In
\secref{optimal}, we present the new optimal construction of
\SSPDs. In \secref{applications}, we present the new construction of a
spanner with a small separator.  We depart with a few concluding
remarks in \secref{conclusions}.

\section{Preliminaries}
\seclab{preliminaries}

\paragraph{Definitions and  notation.}
Let $\BallX{\pnt}{r}$ denote the closed ball centered at a point
$\pnt$ of radius $r$.  Let $\RingX{\pnt}{r}{R} = \BallX{\pnt}{R}
\setminus \BallX{\pnt}{r}$ denote the ring centered at $\pnt$ with
outer radius $R$ and inner radius $r$.

For two sets of points $\PntSet$ and $\PntSetA$ in $\Re^d$, we denote
the distance between the sets $\PntSet$ and $\PntSetA$ by $\ds
\DistSetX{\PntSet}{\PntSetA} = \min_{\pnt \in \PntSet, \pntA \in
   \PntSetA} \dist{\pnt -\pntA}$.  We also use
$\ds \PntSet \otimes \PntSetA = \brc{\brc{\pnt,\pntA} \sep{ \pnt\in
      \PntSet,\, \pntA\in \PntSetA \, \text{ and } \pnt \neq \pntA }}
$
to denote all the (unordered) pairs of points formed by the sets $\PntSet$ and
$\PntSetA$.

\begin{defn}[Pair decomposition.]
    For a point-set $\PntSet$, a \emphi{pair decomposition} of
    $\PntSet$ is a set of pairs $ \displaystyle \WSPDRep =
    \brc{\MakeSBig
       \brc{\SetA_1,\SetB_1},\ldots,\brc{\SetA_s,\SetB_s}}$, such that
   \begin{compactenumi}
       \item $\SetA_i,\SetB_i\subset \PntSet$ for every $i$,
       \item $\SetA_i \cap \SetB_i = \emptyset$ for every $i$, and
       \item $\cup_{i=1}^s \SetA_i \otimes \SetB_i = \PntSet \otimes \PntSet$.
   \end{compactenumi}

   \medskip

   The \emphi{weight} of a pair decomposition $\WSPDRep$ is defined to
   be $\WeightX{\WSPDRep} = \sum_{i=1}^s \pth{ \cardin{\SetA_i } +
      \cardin{\SetB_i}}$.
\end{defn}

As mentioned in the introduction, such a pair decomposition is usually
described implicitly by reporting pairs $\brc{u,v}$ where $u$ and $v$
are nodes of a tree $\QTree$ that stores the points of the given point
set in the leaves and is used to construct the decomposition. As such,
a pair $\brc{u,v}$ represents the pairs $\PSetX{u} \otimes \PSetX{v}$,
where $\PSetX{u}$ is the set of points stored in the subtree of $u$.

\begin{defn}
    A pair of sets $\PntSetB$ and $\PntSetC$ is
    \emphic{$(1/\eps)$-separated}{separated!sets} if $\max \pth{
       \DiamX{\PntSetB}, \DiamX{\PntSetC} } \leq \eps \cdot
    \DistSetX{\PntSetB}{\PntSetC}$. Furthermore, they are
    \emphi{$(1/\eps)$-semi-separated} if $\min \pth{ \DiamX{\PntSetB},
       \DiamX{\PntSetC} } \leq \eps \cdot
    \DistSetX{\PntSetB}{\PntSetC}$.

    \deflab{separation}
\end{defn}

\begin{defn}[\WSPD]
    For a point-set $\PntSet$, a \emphic{well-separated pair
       decomposition}{WSPD} of $\PntSet$ with parameter $1/\eps$ is a
    pair decomposition $\WSPDRep$ of $\PntSet$, such that for any pair
    $\brc{\SetA_i, \SetB_i} \in \WSPDRep$, the sets $\SetA_i$ and
    $\SetB_i$ are $(1/\eps)$-separated.

   \deflab{WSPD}
\end{defn}

\begin{defn}[\SSPD]
   For a point-set $\PntSet$, a \emphic{semi-separated pair
      decomposition}{SSPD} of $\PntSet$ with parameter $1/\eps$,
   denoted by $(1/\eps)$-\SSPD, is a pair decomposition of $\PntSet$
   formed by a set of pairs $\SW$, such that all the pairs of
   $\SW$ are $(1/\eps)$-semi-separated.
\end{defn}

Note that, by definition, a $(1/\eps)$-\WSPD of $\PntSet$ is also a
$(1/\eps)$-\SSPD for $\PntSet$.  Interestingly, one can split any pair
decomposition such that it covers only pairs that appear in some
desired cut.

\begin{lemma}
    \lemlab{split:PD}%
    Given any pair decomposition $\WSPDRep$ of a point-set $\PntSet$,
    and given a subset $\PntSetA$, one can compute a pair
    decomposition $\WSPDRep'$ for $\PntSetA \otimes
    \overline{\PntSetA}$ (that covers \emph{only} these pairs), where
    $\overline{\PntSetA}= \PntSet \setminus \PntSetA$. Furthermore,
    the following properties hold.
    \begin{compactenumA}
        \item If a point appears in $k$ pairs of $\WSPDRep$ then it
        appears in at most $k$ pairs of $\WSPDRep'$.
        \item We have $\WeightX{\WSPDRep'} \leq \WeightX{\WSPDRep}$.
        \item The number of pairs in $\WSPDRep'$ is at most twice
        the number of pairs in $\WSPDRep$.
        \item For any pair $\brc{\SetA', \SetB'} \in \WSPDRep'$ there
        exists a pair $\brc{\SetA, \SetB} \in \WSPDRep$ such that
        $\SetA' \subseteq \SetA$ and $\SetB' \subseteq \SetB$.
        \item The time to compute $\WSPDRep'$ is linear in the weight
        of $\WSPDRep$.
    \end{compactenumA}
\end{lemma}
\begin{proof}
    Let $\WSPDRep' = \brc{ \brc{\SetA \cap \PntSetA, \SetB \cap
          \overline{\PntSetA}} \; , \; \brc{\SetA \cap
          \overline{\PntSetA}, \SetB \cap \PntSetA}\sep{ \brc{\SetA,
             \SetB} \in \WSPDRep}}$. Naturally, we can throw away
    pairs $\brc{X,Y} \in \WSPDRep'$ such that $X$ or $Y$ are empty
    sets.  It is now easy to check that the above properties hold for
    the pair decomposition $\WSPDRep'$.
\end{proof}

We  need the following easy partition lemma.

\begin{lemma}[\cite{hm-fcnld-06}]
    Given any set $\PntSet$ of $n$ points in $\Re^d$, an any constant
    $\mu \geq 1$, and a sufficiently large constant $c \geq 1$ (that
    depends only on $\mu$ and the dimension $d$), one can compute, in
    expected linear time, a ball $\BallX{\pnt}{r}$ that contains at
    least $n/c$ points of $\PntSet$, such that $\BallX{\pnt}{\mu r}$
    contains at most $n/2$ points of $\PntSet$, where $\pnt \in
    \PntSet$.

    \lemlab{small:ball}
\end{lemma}
\begin{proof}
    We include the proof for the sake of completeness.
    Pick randomly a point $\pnt$ from $\PntSet$, and compute the ball
    $\BallX{\pnt}{r}$ of smallest radius around $\pnt$ containing at
    least $n/c$ points of $\PntSet$. Next, consider the ball of radius
    $\BallX{\pnt}{ \mu r}$. If it contains at most $n/2$ points of
    $\PntSet$, then we are done. Otherwise, we repeat this procedure
    until we succeed.

    To see why this algorithm succeeds with constant probability in
    each iteration, consider the smallest radius ball that contains at
    least $m=n/c$ points of $\PntSet$ and is centered at a point of
    $\PntSet$.  Let $\pntA \in \PntSet$ be its center, $\ropt$ be its
    radius, and let $\PntSetA = \PntSet \cap \BallX{\pntA}{\ropt}$ be
    the points of $\PntSet$ contained in this ball. Observe that any
    ball of radius $\ropt / 2$ contains less than $m$ points (the ball
    is not necessarily centered at a point of $\PntSet$).  With
    probability $\geq 1/c$, we have that $\pnt$ is in $\PntSetA$; if
    this is the case, then $r \leq 2\ropt$.

    Furthermore, the ball $\BallX{ \pnt}{ \mu \ropt }$ can be covered
    by $O(1)$ balls of radius $\ropt/2$. Indeed, consider the axis
    parallel box of sidelength $2\mu \ropt$ centered at $\pnt$, and
    partition it into a grid with sidelength $\ropt/\sqrt{d}$. Observe
    that every grid cell can be covered by a ball of radius $\ropt/2$,
    and this grid has at most $c' = (2\mu\ropt/(\ropt/\sqrt{d}) + 1)^d
    = \pth{2 \mu \sqrt{d} + 1}^d$ cells.  Hence it holds that
    $\cardin{\PntSet \cap \BallX{\pnt}{ \mu r}} < c' m \leq n/2$, by
    requiring that $c \geq 2c'$.

    Thus, the algorithm succeeds with probability $1/c$ in each
    iteration, and the expected number of iterations performed is
    $O(c)$. This implies the result, as each iteration takes $O(n)$
    time.
\end{proof}%

The following partition lemma is one of the key ingredients in the new
\SSPD construction.

\begin{lemma}
    Let $\PntSet$ be a set of $n$ points in $\Re^d$, $t > 0$ be a
    parameter, and let $c$ be a sufficiently large constant. Then one
    can compute in linear time a ball $\BallC = \BallX{\pnt}{r}$, such that
    \begin{compactenumi}
        \item $\cardin{\BallC \cap \PntSet} \geq n/c$,
        \item $\cardin{\RingX{\pnt}{r}{r(1+1/t)} \cap \PntSet} \leq
        n/2t$, and
        \item $\cardin{\PntSet \setminus \BallX{\pnt}{2r}} \geq n/2$.
    \end{compactenumi}

    \lemlab{ring:separator}
\end{lemma}
\begin{proof}
    Let $\BallC = \BallX{\pnt}{\alpha}$ be the ball computed, in
    $O(n)$ time, by \lemref{small:ball} such that
    $\cardin{\BallX{\pnt}{\alpha} \cap \PntSet} \geq n/c$ and
    $\cardin{\BallX{\pnt}{8\alpha} \cap \PntSet} \leq n/2$.  We will
    set $r\in[\alpha,e\alpha]$ in such a way that property (ii) will
    hold for it. Indeed, set $r_i = \alpha(1+1/t)^i$, for $i=0,\ldots,
    t$, and consider the rings $\ring_i = \RingX{\pnt}{r_{i-1}}{r_{i}
    }$, for $i=1, \ldots, t$. We have that $r_{t} =\alpha(1+1/t)^{t}
    \leq \alpha \pth{\exp(1/t)}^{t} \leq \alpha e$, since $1+x \leq
    e^x$ for all $x \geq 0$. As such, all these (interior disjoint)
    rings are contained inside $\BallX{\pnt}{4\alpha}$. It follows
    that one of these rings, say the $i$\th ring $\ring_i$, contains
    at most $(n/2)/t$ of the points of $\PntSet$ (since
    $\BallX{\pnt}{8\alpha}$ contains at most half of the points of
    $\PntSet$). For $r = r_{i-1} \leq 4\alpha$ the ball $\BallC =
    \BallX{\pnt}{r}$ has the required properties, as $\BallX{\pnt}{2r}
    \subseteq \BallX{\pnt}{8\alpha}$.
\end{proof}

\subsection{\WSPD construction for the bounded spread case}

The \emphi{spread} of a point-set $\PntSet$ is the quantity $\ds
\SpreadX{\PntSet} = \pth{\max_{\pnt,\pntA \in \PntSet}
   \distX{\pnt}{\pntA} } / \pth{ \min_{\pnt,\pntA \in \PntSet, \pnt\ne
      \pntA} \distX{\pnt}{\pntA}}$.  We next describe a simple \WSPD
construction whose weight depends on the spread of the point set.  The
idea is to use a regular quadtree of height $O(\log \Spread)$ to
compute the \WSPD, and observe that each node participates in
$O(1/\eps^d)$ pairs.

\begin{lemma}
    Let $\PntSet$ be a set of $n$ points in $\Re^d$, with spread $
    \Spread = \SpreadX{\PntSet}$, and let $\eps > 0$ be a
    parameter. Then, one can compute a $(1/\eps)$-\WSPD (and thus a
    $(1/\eps)$-\SSPD) for $\PntSet$ of total weight $O(n \eps^{-d}
    \log \Spread)$. Furthermore, any point of $\PntSet$ participates
    in at most $O\pth{ \eps^{-d} \log \Spread}$ pairs.

    \lemlab{s:s:p:d:spread}
\end{lemma}
\begin{proof}
    Build a regular (i.e., not compressed) quadtree for $\PntSet$ and
    observe that its height is $h = O(\log \Spread)$. To avoid some
    minor technicalities, if a leaf contains a point of $\PntSet$ and
    its height $<h$ then continue refining it till it is of height
    $h$.  Now, all the leafs of the quadtree that have a point of
    $\PntSet$ stored in them are of the same height $h$. Clearly, the
    time to construct this quadtree is $O\pth{ n \log \Spread}$

    Now, construct a $(1/\eps)$-\WSPD for $\PntSet$ using this
    quadtree, such that a node participates only in pairs with nodes
    that are of the exact same level in the quadtree.  This requires a
    slight modification of the algorithm of \cite{ck-dmpsa-95}, such
    that if it fails to separate the pair of nodes $\brc{u,v}$, then
    it recursively tries to separate all the children of $u$ from all
    the children of $v$, thus keeping the invariant that all the pairs
    considered involve nodes of the same level of the quadtree.

    We claim that every node participates in $O(\eps^{-d})$ pairs in
    the resulting \WSPD.  Indeed, a \WSPD pair $\brc{u,v}$ corresponds
    to two cells $\Cell_u$ and $\Cell_v$ that are cubes of the same
    size, such that $\DiamX{\Cell_u}/\eps \leq
    \DistSetX{\Cell_u}{\Cell_v} $. However, for such a pair, we must
    also have that $\DistSetX{\Cell_u}{\Cell_v} \leq 4d \cdot
    \DiamX{\Cell_u}/\eps$. Otherwise, the parents of $u$ and $v$ would
    be $(1/\eps)$-separated, and the algorithm would use them as a
    pair instead of $\brc{u, v}$.

    As such, all the pairs involving a node $u$ in the quadtree must
    use nodes in the same level of the quadtree, and they are in a
    ring of radius $\Theta( \DiamX{\Cell_u}/\eps)$ around it. Clearly,
    there are $O(1/\eps^d)$ such nodes. Implying that $u$ participates
    in at most $O(1/\eps^d)$ pairs.

    Now, a point $\pnt \in \PntSet$ participates in pairs involving
    nodes $v$, such that $v$ is on the path of $\pnt$ from its leaf to
    the root of the quadtree. As such, there are $O( \log \Spread)$
    such nodes, and $\pnt$ will appear in $U = O\pth{ \eps^{-d} \log
       \Spread}$ pairs in the \WSPD. This implies that the total
    weight of this \WSPD is $n U = O\pth{ n\eps^{-d} \log \Spread}$,
    as claimed.
\end{proof}

\begin{remark}
    \lemref{s:s:p:d:spread} implies the desired result (i.e., an \SSPD
    with low weight) if the spread of $\PntSet$ is polynomial in
    $\cardin{\PntSet}$.

    \remlab{low:spread:easy}
\end{remark}

\begin{remark}
    In \lemref{s:s:p:d:spread}, by forcing the $(1/\eps)$-\WSPD to
    form pairs only between nodes in the same level of the quadtree,
    the resulting \WSPD has the property that each node of the
    quadtree participates in $O(1/\eps^d)$ pairs.

    \remlab{few:pairs:per:node}
\end{remark}

\section{A simple construction of \SSPDs}
\seclab{simple}

We first describe a simple construction of \SSPDs for a point-set
$\PntSet$, which is suboptimal.

\begin{theorem}
    Let $\PntSet$ be a set of $n$ points in $\Re^d$, and let $\eps >
    0$ be a parameter. Then, one can compute a $(1/\eps)$-\SSPD for
    $\PntSet$ of total weight $O(n \eps^{-d} \log^2 n)$.

    \lemlab{easy}
\end{theorem}
\begin{proof}
    Using \lemref{ring:separator}, with $t=n$, we compute a ball
    $\BallX{\pnt}{r}$ that contains at least $n/c$ points of
    $\PntSet$,
    and such that $\RingX{\pnt}{r}{(1+1/2t)r}$ contains at most
    $\floor{ n/2t} = \floor{1/2} = 0$ points of $\PntSet$ (that is,
    this ring contains no point of $\PntSet$).

    \begin{figure}[t]%
        \centering
        \includegraphics[scale=0.9999]{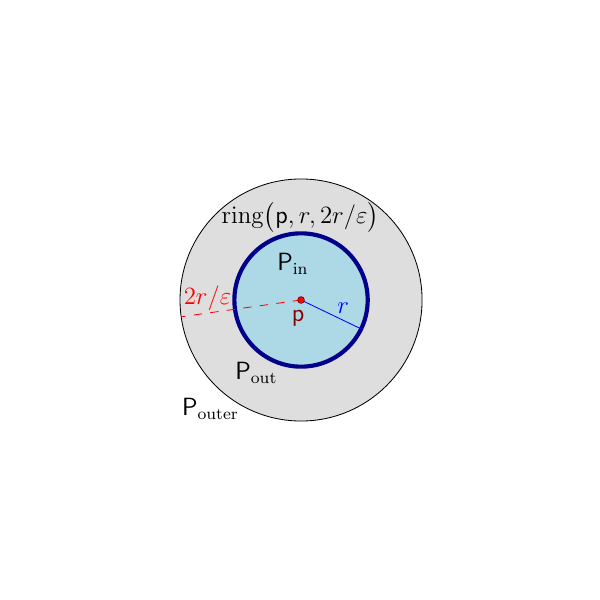}%
        \caption{}
        \figlab{d_s}
    \end{figure}

    Let $\Pin = \PntSet \cap \BallX{\pnt}{r}$, $\Pout = \PntSet \cap \RingX{\pnt}{r}{ 2r/\eps}$, and $\Pouter = \PntSet \setminus \pth[]{ \Pin \cup \Pout}$, see \figref{d_s}.  Clearly, $\brc{\Pin, \Pouter}$ is a $(1/\eps)$-semi-separated pair, which we add to our \SSPD. Let $\ell = \DistSetX{\Pin}{\Pout}$ and observe that $\ell \geq r/2n$.

    We would like to compute the \SSPD for all pairs of points in
    $X = \Pin \otimes \Pout$. Observe that none of these pairs has
    length smaller than $\ell$, and the diameter of the point-set
    $\PntSetA = \Pin \cup \Pout$ is
    $\DiamX{ \PntSetA } \leq 4 \ell n/\eps$. Thus, we can snap the
    point-set $\PntSetA$ to a grid of sidelength
    $\eps \ell / 2\sqrt{d}$. The resulting point-set $\PntSetA'$ has
    spread $O( n/\eps^2)$. Next, compute a $(2/\eps)$-\SSPD for the
    snapped point-set $\PntSetA'$, using the algorithm of
    \lemref{s:s:p:d:spread}. Clearly, the computed \SSPD when
    interpreted on the original point-set $\PntSetA$ would cover all
    the pairs of $X$, and it would provide an $(1/\eps)$-\SSPD for
    these pairs.  By \lemref{s:s:p:d:spread}, every point of
    $\PntSetA$ would participate in at most
    $O\pth{ \eps^{-d} \log(n/\eps)}=O\pth{ \eps^{-d} \log n }$ pairs.

    To complete the construction, we need to construct a
    $(1/\eps)$-\SSPD for the pairs $\Pin \otimes \Pin$ and
    $(\Pout \cup \Pouter) \otimes (\Pout \cup \Pouter)$. To this end,
    we continue the construction recursively on the point-sets $\Pin$
    and $\Pout \cup \Pouter$.

    In the resulting \SSPD, every point participates in at most %
    \begin{equation*}
        T(n)
        = 1 + O\pth{ \eps^{-d} \log n } + \max \pth{ T(n_1), T(n_2)},
    \end{equation*}
    where $n_1 = \cardin{\Pin}$ and $n_2 = \cardin{\Pout \cup
       \Pouter}$.  Since $n_1+ n_2 = n$ and $n_1, n_2\geq n/c$, where
    $c$ is some constant. It follows that $T(n) = O\pth{ \eps^{-d}
       \log^2 n}$.  Now, a point participates in at most $T(n)$ pairs,
    and as such the total weight of the \SSPD is $O(n T(n))$, as
    claimed.
\end{proof}

\section{An optimal construction of \SSPDs}
\seclab{optimal}

\subsection{The construction}
Let $\PntSet$ be a set of $n$ points in $\Re^d$.  If $n = O(1/\eps^d)$
then we compute a $(1/\eps)$-\WSPD of the point set and return it as
the \SSPD.  Otherwise, we compute a ball $\BallX{\pnt}{ r }$ that
contains at least $n/c$ points of $\PntSet$ and such that $\PntSet
\setminus \BallX{\pnt}{ 20r}$ contains at least $n/2$ of the points of
$\PntSet$, where $c$ is a sufficiently large constant that depends
only on the dimension $d$.

\begin{wrapfigure}{r}{0.3\linewidth}
    \centering
    \includegraphics[scale=0.7]{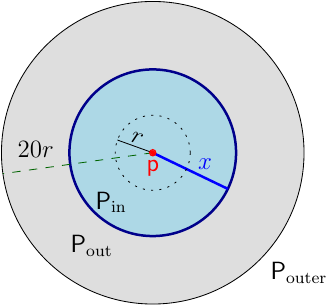}
\end{wrapfigure}
We randomly choose a number $x$ in the range $[5r, 6r]$, and consider the sets:
\begin{align*} \Pin = \PntSet \cap \BallX{\pnt}{ x}, %
  \qquad %
  &%
    \Pout = \PntSet \cap \RingX{\pnt}{x}{20r},%
  \\
  & \text{and} \qquad%
    \Pouter = \PntSet \setminus \pth[]{ \MakeSBig \Pin \cup \Pout}.
\end{align*}
We recursively compute an \SSPD for the set $\Pin$ and an \SSPD for the set $\Pout \cup \Pouter$.

It remains to separate all the pairs of points in $\Pin \otimes
\pth{\Pout \cup \Pouter}$. We do this in two steps, and merge all
these pair decompositions together to get the desired \SSPD of
$\PntSet$.

\subsubsection{Separating \TPDF{$\Pin$}{Pin} from \TPDF{$\Pouter$}{P outer}}

Partition the points of $\Pin$ into $O(1/\eps^d)$ clusters, such that
each cluster has diameter $\leq \eps r/20$. Clearly, we need
$m=O(1/\eps^d)$ such clusters $C_1'',\ldots, C_m''$. Now, since
$\DistSetX{C_i''}{\Pouter} \geq r$, it follows that $C_i''$ and
$\Pouter$ are $(1/\eps)$-semi-separated, for all $i$. Therefore, we
create the pair separating $C_i \otimes \Pouter$, for all $i$. Each
point of $\PntSet$ participates in $O(1/\eps^d)$ such pairs.

We will refer to all the pairs generated in this stage as being
\emphi{long} pairs.

\subsubsection{Separating \TPDF{$\Pin$}{pin} from \TPDF{$\Pout$}{P out}}

This is the more challenging partition to implement as there is no gap
between the two sets.  We build a quadtree $\QTree$ for $\PntSetB =
\Pin \cup \Pout = \PntSet \cap \BallX{\pnt}{20r}$ and compute a
$1/\rho$-\WSPD on this quadtree, where $\rho=\eps/4$.  Namely, a pair
in this decomposition is a pair of two nodes $u$ and $v$ in the
quadtree $\QTree$, such that $\DiamX{\Cell_u} = \DiamX{\Cell_v} \leq
(1 / \rho) \DistSetX{ \Cell_u}{\Cell_v}$, where $\Cell_u$ and
$\Cell_v$ denote the cells in $\Re^d$ that $u$ and $v$ corresponds to.
The construction outputs only the pairs $\{ u,v\}$ such that
$\PSetX{u} \otimes \PSetX{v}$ contains at least one pair of $\Pin
\otimes \Pout$ (i.e., the other pairs of this \WSPD are being
ignored).  The details of how to do this efficiently are described
next.

We  refer to all the pairs generated in this stage as being
\emphi{short} pairs.

\paragraph{On the fly computation of the quadtree.}
To do the above efficiently we do not compute the quadtree $\QTree$ in
advance.  Rather, we start with a root node containing all the points
of $\PntSetB$.  Whenever the algorithm for computing the pair
decomposition tries to access a child of a node $v$, such that $v$
exists in the tree but not its children, we compute the children of
$v$, and split the points currently stored in $v$ (i.e., $\PSetX{v}$)
into the children of $v$. For such newly created child $w$, we check
if $\PSetX{w} \subseteq \Pin$ or $\PSetX{w} \subseteq \Pout$ (and if
so, we turn on the relevant flags in $w$).  Then the regular execution
of the algorithm for computing a pair decomposition resumes.

\paragraph{On the fly pruning of pairs considered.}
Whenever the algorithm handles a pair of vertices $u,v$ of $\QTree$,
it first checks if $\PSetX{u}, \PSetX{v} \subseteq \Pin$ or
$\PSetX{u}, \PSetX{v} \subseteq \Pout$, and if so, the algorithm
returns immediately without generating any pair (i.e., all the pairs
that can be generated from this recursive call do not separate points
of $\Pin$ from $\Pout$ and as such they are not relevant for the task
at hand).  Using the precomputed flags at the nodes of $\QTree$ this
check can be done in constant time.  (A similar idea of doing on the
fly implicit construction of a quadtree while generating some subset
of a pair decomposition was used by Har-Peled \cite{h-pacdp-01}.)

\paragraph{Generating pairs that belong to the same level.}
Note that since we are using a \emph{regular} quadtree to compute the
pair decomposition, we can guarantee that any pair of nodes
$\brc{u,v}$ realizing a pair (i.e., $\PSetX{u} \otimes \PSetX{v}$)
belongs to the same level of the quadtree. Namely, the sidelength of
the two cells $\Cell_u$ and $\Cell_v$ is the same.  (It is easy to
verify that using a regular quadtree to construct \WSPD, instead of
compressed quadtrees, increases the number of pairs in the \WSPD only
by a constant factor.)

\paragraph{Cleaning up.}
Finally, we need to guarantee that only pairs in $\Pin \otimes \Pout$
are covered by this generated (partial) pair decomposition. To this
end, we use the splitting algorithm described in the proof of
\lemref{split:PD} to get this property.

\subsection{Analysis}

\subsubsection{Separating \TPDF{$\Pin$}{Pin}  from
   \TPDF{$\Pout$}{Pout}}

We first analyze the weight of the pairs in the \SSPD separating
$\Pin$ from $\Pout$. Let $\PntSetB = \PntSet \cap \BallX{\pnt}{ 20r}$,
and for a point $\pntA \in \PntSetB$, let $D_\pntA$ be the random
variable which is the (signed) distance of $\pntA$ from the boundary
of the ball $\BallX{\pnt}{x}$. Formally, we have
\[
D_\pntA =  {\dist{\pnt - \pntA} - x }.
\]
For each specific point $\pntA$, the random variable $D_\pntA$ is
uniformly distributed in an interval of length $r$ (note, that for
different points this interval is different).

\begin{claim}
    Consider a point $\pntA \in \PntSetB$, and a pair of nodes of the
    quadtree $\brc{u,v}$ that is in the generated \SSPD of $\Pin
    \otimes \Pout$ and such that $\pntA \in \PSetX{u}$. Then
    $\ds
    \depthX{u} = \depthX{v} \in
    \pbrcx{\lg \frac{1}{\eps}, \;\; \lg \frac{1}{\eps} + \beta +
       \lg \frac{r}{\cardin{D_\pntA}}}$,
    where $\beta$ is some constant. The level $\depthX{u}$ is
    \emphi{active} for $\pntA$, and the number of active levels in the
    quadtree for $\pntA$ is bounded by $\nu(\pntA) = 1 + \beta + \lg
    \frac{r}{\cardin{D_\pntA}}$.

    \clmlab{first}
\end{claim}
\begin{proof}
    Assume that $\pntA \in \Pin$ (a symmetric argument would work for
    the case that $\pntA \in \Pout$), and observe that $\pntA$ is in
    distance at least $\cardin{D_\pntA}$ from all the points of
    $\Pout$.  Consider any pair of nodes of the quadtree $u$ and $v$
    such that:
    \begin{compactenumi}
        \item $u$ and $v$ are in the same level in the  quadtree,
        \item $\pntA \in \PSetX{u}$,
        \item the cells of $u$ and $v$ have diameter
        $\Delta=\DiamX{\Cell_u} =\DiamX{\Cell_v} $, such that $\Delta
        \leq \eps \cardin{D_\pntA}/c$ (for $c$ to be specified
        shortly), and
        \item $\PSetX{v} \cap \Pout \ne \emptyset$.
    \end{compactenumi}
    Now, $\Cell_u$ and $\Cell_v$ have the same side
    length, and the distance between the two cells is at least
    \begin{align*}
        \DistSetX{ \Cell_u}{\Cell_v} &\geq \DistSetX{\pntA}{\Re^d
           \setminus \BallX{\pnt}{x}} - 2\Delta \geq \cardin{D_\pntA}
        - \frac{2\eps \cardin{D_\pntA}}{c}%
        = %
        \pth{1-\frac{2\eps}{c}} \frac{c}{\eps} \cdot \frac{\eps}{c}
        \cardin{D_\pntA}
        \geq %
        \frac{c-2\eps}{\eps} \Delta.
    \end{align*}
    Namely, $u$ and $v$ are $\alpha$--separated, for $\alpha
    =\pth{c-2\eps}/ \eps$.  In particular, by picking $c$ to be
    sufficiently large, we can guarantee that their respective parents
    $\parentX{u}$ and $\parentX{v}$ are $\alpha/4$-separated, and
    $\alpha/4 \geq 1 /\rho$. This implies that $\parentX{u}$ and
    $\parentX{v}$ are $1/\rho$-separated, and the algorithm would have
    included $\brc{\parentX{u}, \parentX{v}}$ in the \SSPD, never
    generating the pair $\brc{u,v}$.

    Namely, a node $u$ in the quadtree that contains $\pntA$ and also
    participates in an \SSPD pair, has $\DiamX{\Cell_u} \geq \eps
    \cardin{D_\pntA}/c$. As such, the node $u$ has depth at most
    \[
    O(1) + \lg \frac{\DiamX{\Pin \cup \Pout}}{\eps \cardin{D_\pntA}/c}
    = O(1) + \lg \frac{40r }{\eps \cardin{D_\pntA}/c} = O(1)+ \lg
    \frac{1}{\eps} + \lg \frac{r}{ \cardin{D_\pntA}},
    \]
    as claimed.

    As for the lower bound -- let $\DiamX{\PntSetB}$ denote the
    diameter of $\DiamX{\PntSetB}$. To get a $(1/\rho)$-separation of
    any two points of $\PntSetB$, one needs to use cells with diameter
    $\leq \rho \,\DiamX{\PntSetB}$. The depth of such nodes in the
    quadtree is $\ds \geq \ceil{\lg \frac{\DiamX{\PntSetB}}{{\rho \,
             \DiamX{\PntSetB}} }} \geq \lg \frac{1}{\eps}$, as
    claimed.
\end{proof}

\begin{claim}
    For a point $\pntA \in \PntSetB$ we have that $\Ex{ \lg \ds
       \frac{r}{\cardin{D_\pntA}}} = O(1)$.

    \clmlab{second}
\end{claim}
\begin{proof}
    Let $I$ be the interval of length $r$, such that $D_\pntA$ is
    distributed uniformly in $I$. Clearly $i \leq \lg \frac{r}{
       \cardin{D_\pntA}} \leq i+1$ if and only if $r2^{-i-1} \leq
    \cardin{D_\pntA} \leq r2^{-i}$. Namely, $D_\pntA \in J_i$ or
    $D_\pntA \in - J_{i}$, where $J_i = \pbrc{ r2^{-i-1} , r2^{-i} }$.
    Therefore,
    \begin{align}
        \Prob{ i \leq \lg \frac{r}{\cardin{D_\pntA}} \leq i+1 } \leq
        \frac{2\cardin{J_i}}{r} \leq \frac{1}{2^i},
        \eqlab{geometric}
    \end{align}
    and $
    \ds \Ex{ \lg \frac{r}{\cardin{D_\pntA}} } \leq \sum_{i=0}^\infty
    (i+1) \cdot \frac{1}{2^i} = O(1)$.
\end{proof}

\begin{lemma}
    The expected total weight of the pairs of the \SSPD of $\Pin
    \otimes \Pout$ is $O(n /\eps^d)$.

    \lemlab{size:p:in:p:out}
\end{lemma}
\begin{proof}
    By \clmref{first} we have that every point $\pntA \in \PntSetB$
    (in expectation) participates in nodes that have depth in the
    range $\pbrcx{ \lg \frac{1}{\eps}, X + O(1) + \lg
       \frac{1}{\eps}}$, where $X = \lg \frac{r}{\cardin{D_\pntA}}$.
    Now, since every node of the regular quadtree participates in at
    most $O(1/\eps^d)$ \WSPD pairs (in the same level, see
    \remref{few:pairs:per:node}), we conclude that $\pntA$
    participates in $O((1 + X)/\eps^d)$ pairs.  By \clmref{second}, we
    have that $O\pth {\Ex{ \MakeBig (1 + X)/\eps^d } } =
    O(1/\eps^d)$. The claim now follows by summing this up over all
    the points of $\PntSetB$.
\end{proof}

\begin{lemma}
    The expected time to compute pairs of the \SSPD of $\Pin \otimes
    \Pout$ is $O(n/\eps^d)$.

    \lemlab{in:out}
\end{lemma}
\begin{proof}
    We break the running time analysis into two parts. First, we bound
    the time to compute (the partial) quadtree $\QTree$. Observe that
    this can be charged to the time spend moving the points down the
    quadtree.  Arguing as in \clmref{first}, a point $\pntA \in
    \PntSetB$ is contained in pairs considered by the algorithm with
    maximum depth $\ds Y_\pntA = O(1) + \lg \frac{1}{\eps} + \lg
    \frac{r}{\cardin{D_\pntA}}$. In particular, the maximum depth of
    $\pntA$ in $\QTree$ is $Y_\pntA+1$. Indeed, $\pntA$ is pushed down
    the quadtree only when the algorithm is trying to separate a pair
    of nodes $\{u, v\}$ such that $\pntA \in \PSetX{u}$ or $\pntA \in
    \PSetX{v}$. As such, the expected time to compute $\QTree$ is
    proportional to $\Ex{\sum_{\pntA \in \PntSetB} Y_\pntA}$, which by
    \clmref{second}, and linearity of expectations, is $O(n \log
    (1/\eps))$.

    Secondly, we need to bound the time it takes to generate the pairs
    themselves.  First, observe that the algorithm considers only
    pairs of nodes that are in the same level of the quadtree.  Now,
    for a specific node $u$ of the quadtree $\QTree$ the total number
    of nodes participating in pairs considered (by the algorithm),
    that include $u$ and are in the same level as $u$ is
    $O(1/\eps^d)$.  In particular, we bound by $O(1/\eps^{2d}) =
    O(n/\eps^d)$ the total time spend by the algorithm in handling
    pairs that are in the top $\alpha = \floor{\lg 1/\eps}$ levels of
    the quadtree.

    To bound the remaining work, consider a point $\pntA$ and all the
    recursive calls in the algorithm that consider nodes that contain
    $\pntA$. The total recursive work that $\pntA$ is involved in is
    bounded by $O(Y_\pntA/\eps^d)$. However, by the above, we can
    ignore the work involved by the top $\alpha$ levels.  As such, the
    total work in identifying the generated pairs involving $\pntA$ is
    bounded by $O((Y_\pntA -\alpha)/\eps^d)$. And in expectation, by
    \clmref{second}, this is $O( 1/\eps^d)$, and $O(n/\eps^d)$
    overall.
\end{proof}

\subsubsection{Bounding the total weight}

\begin{lemma}
    In the \SSPD computed, every point participates in $O( \eps^{-d}
    \log n)$ long pairs.

    \lemlab{single:point:long}
\end{lemma}
\begin{proof}
    The depth of the recursion is $O( \log n)$. A point is being sent
    only to a single recursive call. Furthermore, at each level of the
    recursion, a point might participate in at most $O( \eps^{-d})$
    long pairs.
\end{proof}

\begin{lemma}
    In the \SSPD computed, every point participates in $O( \eps^{-d}
    \log n)$ short pairs, both in expectation and with high
    probability.

    \lemlab{single:point:short}
\end{lemma}
\begin{proof}
    Consider a point $\pntA \in \PntSet$, and let $X_k$ be the number
    of short pairs it appears in when considering the subproblem
    containing it in the $k$\th level of the recursion.  By
    \clmref{first} and \clmref{second}, we have that
    \[
    X_k = O\pth{\frac{\nu(\pntA)}{\eps^d}}%
    =%
    O\pth{ \frac{1}{\eps^d} \pth{ {1 + \lg
             \frac{r}{\cardin{D_\pntA}}}} }%
    =%
    O\pth{ \frac{1+Z_k}{\eps^d}},
    \]
    where $Z_k= \lg \frac{r}{\cardin{D_\pntA}}$ is dominated by a
    geometric variable with expectation $O(1)$ (see
    \Eqref{geometric}). Therefore, the number of pairs $\pntA$
    participates in is bounded by a sum of $h$ geometric variables,
    each one with expectation $O(1/\eps^d)$, where $h=O(\log n)$ is
    the depth of the recursion. These variable arise from different
    levels of the recursion, and are as such independent.  Now, there
    are Chernoff type inequalities for such summations that
    immediately imply the claim---see \cite{mr-ra-95}.
\end{proof}

\begin{lemma}
    Given a point-set $\PntSet$ in $\Re^d$, and parameter $\eps > 0$,
    one can compute, in expected time $O(n \eps^{-d} \log n )$, an
    \SSPD of $\PntSet$ of total expected weight $O\pth{n\eps^{-d} \log
       n}$.  Furthermore, every point participates in $O \pth{
       \eps^{-d} \log n}$ pairs with high probability.

    \lemlab{SSPD:compute}
\end{lemma}
\begin{proof}
    The bound on the total weight of the \SSPD generated is implied by
    \lemref{single:point:long} and \lemref{single:point:short}.  As
    for the running time, \lemref{in:out} implies that in expectation
    the divide stage takes $O( n/\eps^d)$ time, and since the two
    subproblems have size which is a constant fraction of $n$, the
    result follows.
\end{proof}

\subsubsection{Reducing the number of pairs}
In the worst case, the above construction would yield $\Omega_\eps(n
\log n)$ pairs (here $\Omega_\eps$ hides constants that depends
polynomially on $\eps$).  Fortunately, one can reduce the number of
pairs generated.  The idea is to merge together pairs of the \SSPD
that are still well separated together.

We need the following technical lemma.
\begin{lemma}
    Let $\eps \leq 1/12$ and $\eps' \leq \eps/6$ be parameters, let
    $X, Y$ be a $(1/\eps)$-separated pair, and let $X_i, Y_i$ be
    $(1/\eps')$-separated pairs, for $i=1,\ldots, k$. Furthermore,
    assume that for all $i$, we have $X \cap X_i \ne \emptyset$ and $Y
    \cap Y_i \ne \emptyset$. Then $\CX = \cup_i X_i$ is
    $(1/4\eps)$-separated from $\CY = \cup_i Y_i$.

    \lemlab{separation}
\end{lemma}
\begin{proof}
    Let $x \in X$ and $y \in Y$ be the pair of points realizing
    $\ell = \DistSetX{X}{Y}$. Similarly, for all $i$, let
    $\ell_i = \DistSetX{X_i}{Y_i}$. By
    \defref{separation}, for all $i$,  we have that
    \begin{compactenumi}
        \item $\DiamX{X} \leq \eps \ell$,
        \item $\DiamX{Y} \leq \eps \ell$,
        \item $\DiamX{X_i} \leq \eps' \ell_i$, and
        \item $\DiamX{Y_i} \leq \eps' \ell_i$.
    \end{compactenumi}

    We also have that $\ell_i \leq \DistSetX{ \MakeSBig X \cap X_i}{Y
       \cap Y_i} \leq \ell + \DiamX{X} + \DiamX{Y} \leq (1+2\eps)
    \ell$. In particular, as $X$ and $X_i$ have a non-empty
    intersection, we have that $X_i$ (resp. $Y_i$) is contained in a
    ball of radius $\DiamX{X}+\DiamX{X_i} \leq r = \eps \ell + \eps'
    \ell_i$ centered at $x$ (resp. $y$). As such, $\CX$ and $\CY$ are
    contained in balls of radius $r$ centered at $x$ and $y$,
    respectively. The distance between these two balls is at least
    $\ell - 2r$, and the diameter of these two balls is $2r$. As such,
    $\CX$ and $\CY$ are $1/\tau$-separated for
    \begin{align*}
        \tau &= \frac{\max\pth{\DiamX{\CX},
              \DiamX{\CY}}}{\DistSetX{\CX}{\CY}}%
        \leq %
        \frac{2r}{\ell - 2r}%
        =%
        \frac{2\eps \ell + 2\eps' \ell_i}{\ell - 2\eps \ell - 2 \eps'
           \ell_i}%
        \leq%
        \frac{2\eps \ell + 2 \eps' (1+2\eps) \ell}{\ell - 2\eps \ell -
           2 \eps' (1+2\eps) \ell}%
        \\
        &=%
        \frac{2\eps + 2 \eps' (1+2\eps) }{1 - 2\eps - 2 \eps'
           (1+2\eps) }%
        \leq %
        \frac{3\eps }{1 - 3\eps } \leq 4\eps,
    \end{align*}
    as $\eps \leq 1/12$ and $\eps' \leq \eps/6$.
\end{proof}

\begin{lemma}
    \lemlab{reduce}%
    Given an $O(1/\eps)$-\SSPD (similar to the one constructed above)
    with total weight
    \begin{equation*}
        O(n\eps^{-d} \log n ),
    \end{equation*}
    one can reduce the number of pairs in the \SSPD to $O(n/\eps^d)$, and this can be done in $O( n\eps^{-d} \log n )$ time.
\end{lemma}
\begin{proof}
    The number of long pairs created is clearly $O(n /\eps^d)$, as the
    divide stage generates $O(1/\eps^d)$ long pairs, and the recursive
    construction stops when the size of the subproblem drops below
    $O(1/\eps^d)$.

    Let $\SW$ be the computed $O(1/\eps)$-\SSPD (we require a slightly
    stronger separation to implement this part).  We are going to
    merge only the short pairs computed by the algorithm.  Observe,
    that the short pairs in $\SW$ are all $O(1/\eps)$-separated (i.e.,
    not only semi-separated). Construct a $O(1/\eps)$-\WSPD $\WSPDRep$
    of the point-set $\PntSet$. For each \emph{short} pair $A\otimes B
    \in \SW$, pick an arbitrary point $\pntA \in A$ and $\pntB \in B$,
    and find the pair $X \otimes Y$ in $\WSPDRep$ such that $\pntA \in
    X$ and $\pntB \in Y$. Associate the pair $A \otimes B$ with the
    pair $X \otimes Y$. Repeat this for all the pairs in $\SW$.

    Given a pair of points $\pntA, \pntB$ finding the pair of the
    \WSPD containing the two points can be done in constant time
    \cite{fms-aeepw-03}\footnote{This result uses hashing and the
       floor function.}. As such, we can compute, in $O((n/\eps^d)
    \log n)$ time, for all the short pairs in the \SSPD $\SW$ its
    associated pair in the \WSPD. Now, take all the \SSPD pairs
    $\brc{\SetA_1, \SetB_1} ,\ldots, \brc{\SetA_k, \SetB_k}$ that are
    associated with a single pair $\brc{X, Y}$ of the \WSPD, such that
    $\SetA_i \cap X \ne \emptyset$ and $\SetB_i \cap Y \ne \emptyset$
    for all $i$. We replace all these short pairs in the \SSPD $\SW$
    by the single pair
    \[
    \CX \otimes \CY \quad%
    \text{ where } \quad \CX = \pth{\bigcup_i \SetA_i} \quad \text{
       and } \quad \CY = \pth{\bigcup_i \SetB_i}.
    \]
    By \lemref{separation}, $\CX$ and $\CY$ are $1/\eps$-separated. As
    such, in the end of this replacement process, we have a
    $1/\eps$-\SSPD that has $O(n/\eps^d)$ long pairs and
    $\cardin{\WSPDRep}$ short pairs, where $\cardin{\WSPDRep}$ is the
    number of pairs in the \WSPD $\WSPDRep$.  Overall, this merge
    process takes time that is proportional to the total weight of the
    original \SSPD.

    Thus, we get a \SSPD that has $O(n/\eps^d)$ pairs overall, and
    clearly this merging process did not increase the total weight of
    the \SSPD.
\end{proof}

\subsection{The results}

Putting the above together we get our main result.

\begin{theorem}
    Given a point-set $\PntSet$ in $\Re^d$, and parameter $\eps > 0$,
    one can compute, in $O(n \eps^{-d} \log n )$ expected time, an
    \SSPD $\SW$ of $\PntSet$, such that:
    \begin{compactenumA}
        \item The total expected weight of the pairs of $\SW$ is
        $O\pth{n\eps^{-d} \log n}$.
        \item Every point of $\PntSet$ participates in $O \pth{
           \eps^{-d} \log n}$ pairs, with high probability.
        \item The total number of pairs in $\SW$ is $O(n/\eps^d)$.
    \end{compactenumA}

    \thmlab{SSPD:compute}
\end{theorem}

One can  extend the result also to an $n$-point metric space with
bounded doubling dimension.
\begin{theorem}
    Let $\PntSet$ be an $n$-point metric space with doubling dimension
    $\dim$, and parameter $\eps > 0$, then one can compute, in $O(n
    \eps^{-O(\dim)} \log n )$ expected time, an \SSPD $\SW$ of $\PntSet$,
    such that:
    \begin{compactenumA}
        \item The total expected weight of the pairs of $\SW$ is
        $O\pth{n\eps^{-O(\dim)} \log n}$.
        \item Every point of $\PntSet$ participates in $O \pth{
           \eps^{-O(\dim)} \log n}$ pairs, with high probability.
    \end{compactenumA}
    Furthermore, by investing an additional $O\pth{ n \eps^{-O(\dim)}
       \log^2 n }$ time one can reduce the number of pairs in $\SW$ to
    $O(n/\eps^d)$.

    \thmlab{SSPD:compute:d:m}
\end{theorem}
\begin{proof}
    This follows by an immediate plug and chug of our algorithm into
    the machinery of Har-Peled and Mendel \cite{hm-fcnld-06}. Observe
    that our algorithm only used distances in the computation. In
    particular, we replace the use of quadtrees by net-trees, see
    \cite{hm-fcnld-06}.

    The deterioration in the running time to reduce the number of
    pairs is caused by the longer time it takes to perform a
    ``pair-location'' query; that is, locating the pair containing a
    specific pair of points takes $O( \log n)$ time instead of
    constant time.
\end{proof}

It seems that by redesigning the algorithm one can remove the extra
$\log$ factor needed to reduce the number of pairs in the resulting
\SSPD.  However, the resulting algorithm is somewhat more involved and
less clear, and we decided to present here the slightly less
efficient variant.

\section{Applications}
\seclab{applications}

\subsection{Immediate applications}
\seclab{immediate:applications}

We now have a near-linear weight \SSPD for any point set from a finite
metric space with low doubling dimension.  We can use this \SSPD in
any application that uses only the \SSPD property, and do not use any
special properties that exist only in Euclidean space.  So, let
$\PntSet$ be an $n$-point metric space with constant doubling
dimension.  By plugging in the above construction we get the following
new results:
\begin{compactenumA}
    \item A $(1+\eps)$-spanner for $\PntSet$ with (hopping) diameter
    $2$ and $O_\eps (n\log n)$ edges \cite{acfs-psspd-09}, where
    $O_\eps(\cdot)$ hides constants that depends polynomially on
    $1/\eps$.

    \item A $(3+\eps)$-spanner of a complete bipartite graph
    \cite{acfs-psspd-09}. (This can also be done directly by using
    \WSPD.)

    \item An additively $(2+\eps)$-spanner with $O_\eps(n\log n)$
    edges \cite{adfgs-gswps-11} can be constructed for $\PntSet$. See
    \cite{adfgs-gswps-11} for exact definitions.
\end{compactenumA}
\smallskip%
The previous results mentioned above were restricted to points lying
in $\Re^d$, while the new results also works in spaces with low
doubling dimension.

\medskip

One can also extract a spanner from any \SSPD, as the following
theorem states. Since the proof of this is standard, we delegate the
proof to \apndref{SSPD:to:spanners}.

\newcommand{\ThmSpannerBody}{%
   Given a $8/\eps$-\SSPD $\SW$ for a point-set $\PntSet$ in $\Re^d$,
   one can compute a $(1+\eps)$-spanner of $\PntSet$ with
   $O(\cardin{\SW}/\eps^{d-1})$ edges. The construction time is
   proportional to the total weight of $\SW$. In particular, a point
   appearing in $k$ pairs of the \SSPD is of degree $O\pth{ k
      /\eps^{d-1}}$ in the resulting spanner.%
}

\begin{theorem}
    \ThmSpannerBody{}

    \thmlab{spanner}
\end{theorem}

\subsection{Spanners with \TPDF{$O(n^{1-1/d})$}{n 1/d}-separator}

Let $\PntSet$ be a set of $n$ points in $\Re^d$, and let $\eps > 0$ be
a parameter. We next describe a modified construction of \SSPDs given
in \secref{simple} for the point-set $\PntSet$, such that when
converting it into a spanner, it has a small separator.

\subsubsection{\SSPD and \WSPD for mildly separated sets}

\begin{lemma}
    \lemlab{mild:W:S:P:D}%
    Let $\PntSet = \Pin \cup \Pout$ be a set of $n$ points in $\Re^d$,
    $\pnt$ a point, and $r, \Thickness, R$ numbers, such that the
    following holds
    \begin{compactenumi}
        \item there is a ring $\ring = \RingX{\pnt}{r}{r+\Thickness}$ that
        separates $\Pin$ from $\Pout$,
        \item $\Pin \subseteq \BallX{\pnt}{r}$, and
        \item $\Pout \subseteq \BallX{\pnt}{R} \setminus
        \BallX{\pnt}{r+\Thickness}$.
    \end{compactenumi}
    Then, for $\eps>0$, one can compute $1/\eps$-\WSPD $\WSPDRep$
    covering $\Pin \otimes \Pout$, such that:
    \begin{compactenumA}
        \item There are $O \pth{ \eps^{-2d} \pth{ \pth{{r} /
                 {t}}^{d-1} + \log \pth{ {R}/{ \ds \eps \Thickness }}
           }}$ pairs in $\WSPDRep$.
        \item Every point participates in $O\pth{ \eps^{-d} \lg (R/
           \eps \Thickness  ) }$ pairs in $\WSPDRep$.
        \item The total weight of $\WSPDRep$ is $O\pth{ n \eps^{-d}
           \lg (R/ \eps \Thickness  ) }$.
        \item The time to compute $\WSPDRep$ is $O\pth{ n \eps^{-d}
           \lg (R/ \eps \Thickness ) }$.
    \end{compactenumA}
\end{lemma}
\begin{proof}
    Snap the point set $\PntSet = \Pin \cup \Pout$ to a grid with
    sidelength $\eps \Thickness/ 8d$, and let $\PntSetC$ denote the
    resulting point set. The set $\PntSetC$ has spread $\Spread =
    O(R/\eps \Thickness)$.  Use the algorithm of
    \lemref{s:s:p:d:spread} to compute a $4/\eps$-\WSPD for
    $\PntSetC$. Now, interpret this \WSPD as being on the original
    point set, and use \lemref{split:PD} to convert it into a \WSPD
    covering only $\Pin \otimes \Pout$. Clearly, this is the required
    $1/\eps$-\WSPD $\WSPDRep$ covering $\Pin \otimes \Pout$. The
    running time of the algorithm is $O( n\eps^{-d} \log \Spread )$.

    We need to bound the number of pairs generated, and their total
    weight.  So, consider the quadtree used in computing the \WSPD for
    $\PntSetC$. Its root has diameter $\Delta_0 \leq 2d R$, and a node
    in the $i$\th level of the quadtree has diameter at most $\Delta_i
    = \Delta_0/2^i$.

    We will refer to a pair of the \WSPD computed for $\PntSetC$ that
    induces at least one non-empty pair in $\WSPDRep$ as
    \emph{active}.  Now, a node $u$ of level $i$ in distance $\ell$
    from the ring, can not participate in an active pair $\brc{u,v}$
    if $\ell > c \Delta_i/\eps$, for a sufficiently large constant
    $c$, since the parents of $u$ and $v$ are already well-separated,
    and $\PSetX{u} \cup \PSetX{v}$ contain points from (the snapped
    version of) both $\Pin$ and $\Pout$.  Observe that a sphere of
    radius $\rho$, can intersect at most
    $O\pth{(\rho/\Delta_i)^{d-1}}$ cells of a grid of sidelength
    $\Delta_i$.  Since for $a, b > 0$, we have that $(a+b)^{d-1} \leq
    2^{d-1}\pth{a^{d-1} + b^{d-1}}$, we conclude that
    $\RingX{\pnt}{r-c \Delta_i/\eps}{r + c\Delta_i/\eps}$ intersects
    at most
    \[
    O\pth{ \pth{\frac{r + \Delta_i/\eps}{\Delta_i}}^{d-1} \frac{ c
          \Delta_i/\eps}{\Delta_i}}%
    =%
    O\pth{  \frac{ 1}{\eps} \pth{ \frac{r }{\Delta_i}}^{d-1} +
          \frac{1}{\eps^{d}}} %
    \]
    nodes of level $i$ that are active (i.e., participate in pairs of
    $\WSPDRep$).  The bottom level of the quadtree that has active
    pairs is (at most) $h = \ceil{\lg_2 (2R d / \eps \Thickness)}$,
    and summing over all levels, we have that the number of active
    nodes overall is
    \begin{align*}
        N &= \sum_{i=0}^{h} O\pth{ \frac{ 1}{\eps} \pth{ \frac{r
              }{\Delta_i}}^{d-1} \!\!\! + \frac{1}{\eps^{d}}} %
        =%
        O\pth{ \frac{ 1}{\eps}\pth{\frac{r }{\eps \Thickness}}^{d-1}
           \!\!\!+ \frac{\log {R}/{\eps \Thickness}}{\eps^d }}
        \\%
        &=%
        O \pth{ \frac{1}{\eps^{d}} \pth{
              \pth{\frac{r}{\Thickness}}^{d-1} \!\!\! + \log
              \frac{R}{\eps \Thickness}}}.
    \end{align*}
    The total number of pairs generated is $O\pth{N/\eps^d}$ since
    every node participates in $O\pth{1/\eps^d}$ pairs. Every point
    participates in $K = O( h /\eps^d ) = O\pth{ \eps^{-d} \log
       (R/\eps \Thickness) }$ pairs, and the total weight of the
    generated \WSPD is $O\pth{ n K } = O\pth{ n \eps^{-d} \log (R/\eps
       \Thickness) }$.
\end{proof}

\begin{lemma}
    Given a point set $\PntSet = \Pin \cup \Pout$, and a
    $\alpha$-\WSPD $\WSPDRep$ covering $\Pin \otimes \Pout$, such that
    $\WSPDRep$ has $N$ pairs, it is of total weight $W$, and every
    point of $\PntSet$ participates in at most $T$ pairs, then one can
    convert it into a $\alpha\eps$-\SSPD with $O( N/\eps^d )$ pairs,
    and of total weight $O( W/ \eps^d)$, such that every point
    participates in $O(T/\eps^d)$ pairs.
\end{lemma}
\begin{proof}
    Given a pair $\brc{X,Y} \in \WSPDRep$, split $X$ into $m =
    O(1/\eps^d)$ clusters $X_1,\ldots, X_m$, each with diameter $\leq
    \eps \DiamX{X}/4$.  Clearly, each such cluster $X_i$ is
    $\alpha\eps$-semi-separated from $Y$. Now, replacing $\brc{X,Y}$
    by the semi-separated pairs $\brc{X_1,Y}, \ldots, \brc{X_m, Y}$,
    and repeating this for all pairs in $\WSPDRep$, yields the
    required \SSPD.
\end{proof}

Applying \lemref{mild:W:S:P:D}, with constant separation, and then
refining it using the above lemma, implies the following.

\begin{lemma}
    \lemlab{mild:S:S:P:D}%
    Let $\PntSet = \Pin \cup \Pout$ be a set of $n$ points in $\Re^d$,
    $\pnt$ a point, and $r, \Thickness, R$ numbers, such that the following
    holds
    \begin{compactenumi}
        \item there is a ring $\ring = \RingX{\pnt}{r}{r+\Thickness}$ that
        separates $\Pin$ from $\Pout$,
        \item $\Pin \subseteq \BallX{\pnt}{r}$, and
        \item $\Pout \subseteq \BallX{\pnt}{R} \setminus
        \BallX{\pnt}{r+\Thickness}$.
    \end{compactenumi}
    Then, for $\eps>0$, one can compute $1/\eps$-\SSPD $\WSPDRep$
    covering $\Pin \otimes \Pout$, such that:
    \begin{compactenumA}
        \item There are $O \pth{ \eps^{-d} \pth{ \pth{{r} /
                 {\Thickness}}^{d-1} + \log (R / \Thickness ) }}$ pairs in
        $\WSPDRep$.
        \item Every point participates in $O\pth{ \eps^{-d} \lg (R/
           \Thickness ) }$ pairs in $\WSPDRep$.
        \item The total weight of $\WSPDRep$ is $O\pth{ n \eps^{-d}
           \lg (R/ \Thickness ) }$.
        \item The time to compute $\WSPDRep$ is $O\pth{ n \eps^{-d}
           \lg (R/ \Thickness ) }$.
    \end{compactenumA}
\end{lemma}

\subsubsection{The \SSPD construction}

Compute a ball $\BallX{\pnt}{r}$, using \lemref{ring:separator}, with
$t= (1/2) n^{1/d}$. Let $r' = (1+1/t)r$. We have that
$\RingX{\pnt}{r}{r'}$ contains at most $n^{1-1/d}$ points of
$\PntSet$. We partition $\PntSet$ as follows:
\begin{align*}
    \Pin &= \PntSet \cap \BallX{\pnt}{r}, \;\;\;\; & \Pring = \PntSet
    \cap
    \RingX{\pnt}{r}{r'},\\
    \Pout &= \PntSet \cap \RingX{\pnt}{r'}{2r}, & \text{ and ~ ~
    }\Pouter = \PntSet \setminus \BallX{\pnt}{2r}.
\end{align*}
We have that $|\Pin| \geq n/c$, $|\Pring| \leq n^{1-1/d}$, and
$|\Pouter| \geq n/2$ where $c$ is a constant that depends only on the
dimension $d$.  We add the following pairs to the \SSPD $\SW$.
\begin{compactenumA}
    \item \label{spanner:A}%
    $\Pring \otimes (\PntSet \setminus \Pring)$: We compute a
    $(1/\eps)$-\SSPD for $\PntSet$ using \thmref{SSPD:compute}, and we
    split it using \lemref{split:PD} to get a $(1/\eps)$-\SSPD for
    $\Pring \otimes (\PntSet \setminus \Pring)$ (with the same
    parameters as the original \SSPD of \thmref{SSPD:compute}).

    \item \label{spanner:B}%
    $\Pin \otimes \Pouter$: We decompose $\Pin$ into
    $O(1/\eps^d)$ clusters each with diameter $\eps r/10$ (for
    example, by using a grid of the appropriate size). Let $\Family$
    be the resulting set of subsets of $\Pin$. For each $X \in
    \Family$, we add $\brc{X,\Pouter}$ as a $(1/\eps)$-semi-separated
    pair to $\SW$.

    \item \label{spanner:C} $\Pin \otimes \Pout$: Compute a
    $O(1/\eps)$-\SSPD for $\Pin \otimes \Pout$ using the algorithm of
    \lemref{mild:S:S:P:D}. The diameter of $\Pin \otimes \Pout$ is
    $4r$, and the ring thickness separating $\Pin$ from $\Pout$ is
    $r/t$. The resulting \SSPD has $O \pth{ \eps^{-d} \pth{ t^{d-1} +
          \log t }}= O\pth{\eps^{-d} n^{1-1/d} }$ pairs, any point
    participates in at most $O(\eps^{-d} \log n )$ pairs, and the
    total weight of the \SSPD is $O(n \eps^{-d} \log n )$.

    \item %
    $\Pin \otimes \Pin$, $\Pring \otimes \Pring$ and $(\Pout \cup
    \Pouter) \otimes (\Pout \cup \Pouter)$: We construct the \SSPD for
    these two sets of pairs by recursively calling the algorithm on
    the sets $\Pin$, $\Pring$ and on $\Pout \bigcup \Pouter$.
\end{compactenumA}

\begin{lemma}
    For a point-set $\PntSet$ of $n$ points in $\Re^d$, and a
    parameter $\eps$, the $(1/\eps)$-\SSPD generated by the above construction
    has:
    \begin{compactenumi}
        \item   $O\pth{n/\eps^d}$ pairs,
        \item every point is contained in $O\pth{ \eps^{-d} \log^2 n
        }$ pairs,
        \item the total weight of the \SSPD is $O\pth{ n \eps^{-d}
           \log^2 n }$, and
        \item the construction time is $O(n \eps^{-d} \log^2 n)$.
    \end{compactenumi}

    \lemlab{SSPD:parameters}
\end{lemma}
\begin{proof}
    (i) Let $T(n)$ denote the number of pairs generated by the above
    algorithm for a set of $n$ points. We bound the number of pairs
    generated in each stage separately:

    \begin{compactenumA}
        \item For $\Pring \otimes (\PntSet \setminus \Pring)$, observe
        that every point participates in at most $O( \eps^{-d} \log
        n)$ pairs, and there are $O(n^{1-1/d})$ points in $\Pring$. As
        such, the number of pairs generated in this step is
        $O\pth{\pth{n^{1-1/d}/\eps^d} \log n}$, as every pair must
        involve at least one point of $\Pring$.  Namely, the total
        number of pairs generated for $\Pring \otimes (\PntSet
        \setminus \Pring)$ is $O\pth{\pth{n^{1-1/d}/\eps^d} \log n}$.

        \item Separating $\Pin \otimes \Pouter$ requires $O(1/\eps^d)$
        pairs.

        \item Separating $\Pin \otimes \Pout$ requires
        $O\pth{\eps^{-d} n^{1-1/d} }$ pairs.

        \item Separating $\Pin \otimes \Pin$, $\Pring \otimes \Pring$
        and $(\Pout \cup \Pouter) \otimes (\Pout \cup \Pouter)$
        requires $T(\cardin{\Pin})$, $T(\cardin{\Pring})$ and
        $T(\cardin{\Pout \cup \Pouter})$ pairs, respectively.
    \end{compactenumA}

    As such, we have that
    \[
     T(n) = O\pth{\pth{n^{1-1/d}/\eps^d} \log
       n} + T\pth{ \cardin{\Pin} } + T\pth{ \cardin{\Pring} } + T\pth{
       \cardin{\Pout \cup \Pouter} }.
    \]
    The solution to this recurrence is $O\pth{n/\eps^d}$.

    \medskip \noindent (ii) We have that a point participates in at
    most
    \begin{align*}
        D(n) &= O( \eps^{-d} \log n) + O( \eps^{-d} \log (n/\eps)) +
        \max\pth{ \MakeBig%
           D\pth{ \MakeSBig\! \cardin{\Pin}}, %
           D\pth{ \MakeSBig\! \cardin{\Pring} }, %
           D\pth{ \MakeSBig\! \cardin{\Pout \cup \Pouter} }}\\
        &= O( \eps^{-d} \log n) + D\pth{\MakeBig (1-1/c) n },
     \end{align*}
     as $\eps \geq 1/n$, $\cardin{\Pin} \leq n/2$, $\cardin{\Pring}
     \leq n^{1-1/d}$, and $\cardin{\Pout \cup \Pouter} \leq n -
     \cardin{\Pin} \leq (1-1/c)n$.  The solution to this recurrence is
     $O\pth{ \eps^{-d} \log^2 n }$.

    \medskip \noindent (iii) By the above, the total weight of the
    \SSPD generated is $O\pth{ n \eps^{-d} \log^2 n }$.

    \medskip \noindent (iv) As for the construction time, we get $R(n)
    = O\pth{n\eps^{-d} \log n} + R(n_1) + R(n_2) + R(n_3)$, where $n_1
    + n_2 +n_3 \leq n$ and $n_1, n_2, n_3 \leq (1-1/c)n$ for some
    absolute constant $c>1$. As such, the construction time is $O(n
    \eps^{-d} \log^2 n)$.
\end{proof}

\subsubsection{Converting into spanner}

We plug the above \SSPD construction into \thmref{spanner} to get a
spanner. We remind the reader that the construction uses a cone
decomposition around the smaller part of each pair in the \SSPD and
connects the apex of the cone to the closest point in the other side
of the pair inside the cone, as such every \SSPD pair give rise to
$O\pth{1/\eps^{d-1}}$ edges, which all share a single vertex called
the \emphi{hub} of the pair.  Let $\Graph$ be the resulting spanner.

\begin{lemma}
    The graph $\Graph$ has a separator of size $O\pth{
       n^{1-1/d}/\eps^{d}}$.

    \lemlab{small:separator}
\end{lemma}

\begin{proof}
    If we remove all the points of $\Pring$, the graph $\Graph$ is
    almost disconnected.  Indeed, removing all $O(n^{1-1/d})$ points
    of $\Pring$ immediately kills all the edges of the spanner that
    rise out of stage (\ref{spanner:A}). Stage (\ref{spanner:B}) gives
    rise to $O(1/\eps^d)$ pairs and consequently $O(1/\eps^{2d-1})$
    edges. Eliminating the $O(1/\eps^d)$ hub vertices eliminates all
    such edges. Stage (\ref{spanner:C}) gives rise to
    $O(n^{1-1/d}/\eps^d)$ pairs, which in turn induces
    $O\pth{n^{1-1/d} /\eps^{2d - 1}}$ edges in the spanner. Again,
    removing the corresponding $O(n^{1-1/d}/\eps^d)$ hub vertices
    eliminates all such edges. Therefore the set $\SetA$, of all these
    $O(n^{1-1/d}/\eps^d)$ removed vertices, is an
    $O(n^{1-1/d}/\eps^d)$-separator for graph $\Graph$ as it separates
    $\Pin$ from $\Pout \cup \Pouter$.
\end{proof}

\begin{theorem}
    For any $\eps>0$ and any set $\PntSet$ of $n$ points in $\Re^d$,
    there is a $(1+\eps)$-spanner $\Graph$ with
    \begin{compactenumi}
        \item $O(n/\eps^{2d-1})$ edges,
        \item maximum degree $O\pth{
           (1/\eps^{2d-1}) \log^2 n }$, and
        \item a separator of size $O(
        n^{1-1/d}/\eps^{d})$.
    \end{compactenumi}
    The $(1+\eps)$-spanner can be constructed in
    $O\pth{(n/\eps^d)\log^{2} n}$ time.

    \thmlab{separator}
\end{theorem}
\begin{proof}
    Computing the \SSPD takes $O\pth{(n/\eps^d)\log^{2} n}$ time, and
    this also bounds the total weight of the \SSPD.
    \thmref{spanner} converts this into the desired spanner in time
    proportional to the total weight of the \SSPD.
\end{proof}

\thmref{separator} compares favorably with the result of \Furer and
Kasiviswanathan \cite{fk-sgig-07}. Indeed, the stated running time of
their algorithm is $O\pth{n^{2-2/(\ceil{d/2} + 1)} }$ (ignoring
polylog factors and the dependency on $\eps$). It is quite plausible
that their algorithm can be made to be faster (most likely
$O(n\log^{d-1} n)$) for the special case of the complete
graph. However, in the worst case, the maximum degree of a vertex in
their spanner is $\Omega(n)$, while in our construction the maximum
degree is $O(\log^2 n)$.

\begin{remark}
    Consider the point set made out of the standard $n^{1/d} \times
    n^{1/d} \times \cdots \times n^{1/d}$ grid. For any $\eps<1$, all
    the edges of the grid must be in the spanner of this
    graph. However, any separator for such a grid graph requires
    $\Omega\pth{ n^{1-1/d} }$ vertices \cite{rh-gsa-01}.  Namely, the
    bound of \thmref{separator} is close to optimal in the worst case.
\end{remark}

\section{Conclusions}
\seclab{conclusions}

We presented several new constructions of \SSPDs that have several
additional properties that previous constructions did not have. Our
basic construction relied on finding a good ring separator in low
dimension, an idea that should have other applications. To get an
optimal construction we used a random partition scheme which might be
of independent interest.

Many of the applications of \SSPDs uses cones and angles. It would be
interesting to extend some of these applications to spaces with low
doubling dimension. In particular, can one construct a spanner for a
point set, with low degree and a small separator, in a
low-doubling-dimension space?

\subsubsection*{Acknowledgments.}

The authors thank Sylvie Temme and Jan Vahrenhold for pointing out
mistakes in earlier versions of this paper. The authors also thank the
anonymous referees for their detailed and insightful comments.

\printbibliography

\appendix

\section{Converting a \SSPD into a spanner}
\apndlab{SSPD:to:spanners}

Abam \etal \cite{adfg-rftgs-09} showed how to construct a
$(1+\eps)$-spanner from their \SSPD for a set of $n$ points in a
plane. The proof that the resulting graph is $(1+\eps)$-spanner
depends on the \SSPD construction, namely the monotonicity property of
the \SSPD---see \cite{adfg-rftgs-09} for the definition and details.
Our spanner construction from a \SSPD $\SW$ of a set $\PntSet
\subseteq \Re^d$ of $n$ points, as describe next, is just a simple
generalization of the construction given in \cite{adfg-rftgs-09} but
the proof is independent of the construction.

\subsection{The construction}

For a parameter $\coneAngle$, we define a \emphi{$\coneAngle$-cone} to
be the intersection of $d$ non-parallel half-spaces such that the
angle of any two rays emanating at the cone's apex and being inside
the cone is at most~$\coneAngle$.  Let $\cset$ be a collection of
$O(1/\coneAngle^{d-1})$ interior-disjoint $\coneAngle$-cones, each
with their apex at the origin, that together cover $\Re^d$.  We can
construct this collection of cones by a grid induced by a system of
halfspaces, such that given a direction, we can find the cone
containing this direction in constant time.

We call the cones in $\cset$ \emphi{canonical cones}.  For a cone
$\cone \in \cset$ and a point $p \in \Re^d$, let $\cone(p)$ denote the
translated copy of $\cone$ whose apex coincides with~$\pnt$.

\subsubsection{The spanner construction}

We are given a $1/\rho$-\SSPD $\SW$ of a point-set $\PntSet$ in
$\Re^d$, where $\rho \leq \eps/8$. We next show how to convert this
\SSPD into a $(1+\eps)$-spanner of $\PntSet$.

Let $\coneAngle = \eps/40$. We build a graph $\Graph$ having $\PntSet$
as its vertex set. Initially the graph has no edges.  Next, for each
pair $\brc{\SetA, \SetB} \in \SW$ pick an arbitrary point $\pnt$ from
the set with smaller diameter, say $\SetA$. We will refer to $\pnt$ as
the \emphi{hub} of $\SetA$, denoted by $\hub(X)$. For each cone $\cone
\in \cset$, we connect $\pnt$ to its nearest neighbor in $\SetB \cap
\cone(\pnt)$ (denoted by $\pntA$); that is, we insert the edge $\pnt
\pntA$ into $\Graph$, with $\distX{\pnt}{\pntA}$ as its weight. Thus,
every pair of $\SW$ contributes $\cardin{\cset}$ edges to $\Graph$.

We claim that $\Graph$ is the desired spanner.

\subsection{Analysis}

\noindent \textbf{Restatement of     \thmref{spanner}}
{\emph \ThmSpannerBody}
\bigskip

\begin{proof}
    The construction is described above (using $\rho = \eps/8$ and
    $\coneAngle = \eps/40$), and the bound on the number of edges in
    the spanner follows immediately.

    The proof of the spanner property is by induction. Sort the pairs
    of points in $\PntSet$ by their length and let $p_1 q_1, \ldots,
    p_{u}q_u$ be these $\ds u = \binom{n}{2}$ sorted pairs (in
    increasing order), where $n = \cardin{\PntSet}$.

    It is easy to verify that $p_1 q_1$ must be in the spanner.
    Assume it holds that $\distG{p_i}{q_i} \leq (1+\eps) \distX{ p_i}{
       q_i}$, for all $i \leq k$. We now prove the claim holds for
    $\src\target$, where $\src = p_{k+1}$ and $\target = q_{k+1}$.

    \begin{figure}[h]%
        \centering%
        \includegraphics{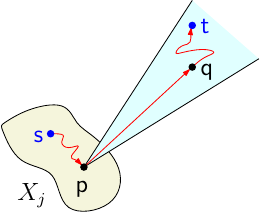}%
    \end{figure}
    Assume $\src \in \SetA_j$ and $\target \in \SetB_j$ for some pair $\brc{\SetA_j, \SetB_j}$ of $\SW$.  Furthermore, assume that $\SetA_j$ is the set with the smaller diameter, and let $\pnt = \hub(\SetA_j)$. Let $\pntA$ be the closest neighbor to $\pnt$ inside the cone containing $\target$. By construction, the edge $\pnt \pntA$ is in the spanner.

    Since $\src, \pnt \in \SetA_j$, and $\target \in \SetB_j$ it follows that $\distX{\src}{\pnt} < \distX{\src}{\target}$. As such, by induction, we have that $\distG{\src}{\pnt} \leq (1+ \eps) \distX{\src}{\pnt}$. It is also easy to verify that $\distX{\pntA}{\target} < \distX{\src}{\target}$ (this also follows from the calculations below), which implies that $\distG{\pntA}{\target} \leq (1+ \eps) \distX{\pntA}{\target}$. As such, we have that
    \begin{align*}
        \Err &=
        \distG{\src}{\target}  - \distX{\src}{\target}%
        \leq
        (1+\eps)\distX{\src}{\pnt} +
        \distX{\pnt}{\pntA} +
        (1+\eps)\distX{\pntA}{\target}  - \distX{\src}{\target}
        \\
        &\leq
        2\distX{\src}{\pnt} +
        \distX{\pnt}{\pntA} +
        (1+\eps)\distX{\pntA}{\target}  - \distX{\pnt}{\target}
        + \distX{\src}{\pnt}
        \leq
        3\distX{\src}{\pnt} +
        \overbrace{\distX{\pnt}{\pntA} +
        (1+\eps)\distX{\pntA}{\target}  - \distX{\pnt}{\target}}^{\Delta=}.
    \end{align*}

    \begin{figure}[h]%
        \centering
        \includegraphics{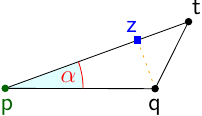}
    \end{figure}
    Let $\pntC$ be the projection of $\pntA$ on the segment $\pnt \target$. It holds $\distX{\pnt}{\target} = \distX{\pnt}{\pntC} + \distX{\pntC}{\target}$.  Let $\alpha = \angle \target \pnt \pntA \leq \coneAngle$.  Now, observe that $\ds \tan \alpha = \frac{\sin \alpha}{\cos \alpha} \leq 2 \coneAngle, $ because $\sin\alpha \leq \alpha \leq \coneAngle$, and $\cos \alpha \geq \cos \coneAngle \geq 1/2$, since $\coneAngle \leq \pi /3$. As such $\distX{\pntC}{\pntA} = \distX{\pnt}{\pntC} \tan \alpha \leq 2 \coneAngle \distX{\pnt}{\pntC} $. This implies that
    \begin{align*} %
        \Delta & = \distX{\pnt}{\pntA} +
        (1+\eps)\distX{\pntA}{\target} - \distX{\pnt}{\target} \leq
        \distX{\pnt}{\pntC} + \distX{\pntC}{\pntA} + (1+\eps)\pth{
           \distX{\pntC}{\pntA} + \distX{\pntC}{\target}} -
        \distX{\pnt}{\pntC} - \distX{\pntC}{\target} \\ %
        & %
        \leq %
        (2+\eps) \distX{\pntC}{\pntA} + \eps \distX{\pntC}{\target}
        \leq%
        2\coneAngle (2+\eps)  \distX{\pnt}{\pntC}
        + \eps \distX{\pntC}{\target}
        \leq%
        6 \coneAngle  \distX{\pnt}{\pntC}
        + \eps \distX{\pntC}{\target}
        \\
        &\leq
        (6 \coneAngle -\eps )  \distX{\pnt}{\pntC}
        + \eps \pth{  \distX{\pnt}{\pntC}  +
           \distX{\pntC}{\target}}%
        =
        (6 \coneAngle -\eps )  \distX{\pnt}{\pntC}
        + \eps  \distX{\pnt}{\target } .
    \end{align*}

    Now, by the above $\distX{\pnt}{\pntC} \leq \distX{\pnt}{\pntA}$
    and $\distX{\src}{\pnt} \leq (\eps/8) \distX{\pnt}{\pntA}$. As
    such,
    \begin{align*}
        \Err %
        &%
        \leq%
        3\distX{\src}{\pnt} + \Delta \leq %
        3\distX{\src}{\pnt} + (6 \coneAngle -\eps )
        \distX{\pnt}{\pntC} + \eps \pth{ \distX{\src}{\pnt}
           +\distX{\src}{\target}} \\
        &\leq%
        4 \distX{\src}{\pnt} + (6 \coneAngle -\eps )
        \distX{\pnt}{\pntC} + \eps \distX{\src}{\target}%
        \leq%
        \frac{\eps}{2} \distX{\pnt}{\pntA} + (6 \coneAngle -\eps )
        \distX{\pnt}{\pntA} + \eps \distX{\src}{\target} \\%
        &=%
        \pth{6 \coneAngle -\frac{\eps}{2} } \distX{\pnt}{\pntA} + \eps
        \distX{\src}{\target} \leq \eps \distX{\src}{\target},
    \end{align*}
    since $\coneAngle \leq \eps/12$ implies that $\pth{6 \coneAngle
       -\frac{\eps}{2} }\leq 0$. This implies that
    $\distG{\src}{\target} \leq (1+\eps) \distX{\src}{\target}$.

    \paragraph{Construction time.} To implement this construction we
    scan the pairs of the \SSPD one by one. For each such pair
    $\brc{\SetA,\SetB}$, we do the following.

    First, we need to determine if $\DiamX{\SetA} = O(\DiamX{\SetB})$
    or $\DiamX{\SetB} = O(\DiamX{\SetA})$ (so we known which set is
    roughly smaller). This can be done in $O(\cardin{\SetA} +
    \cardin{\SetB})$ by approximating the diameter of both sets in
    linear time.  Next, for each such pair (assume $\SetA$ has a
    smaller diameter than $\SetB$), we need to find the nearest
    neighbor to $\hub(\SetA)$ in $\SetB$, in each one of the cones. To
    this end, we have a grid of directions around $\hub(\SetA)$, and
    we compute for each grid cell all the points of $\SetB$ falling
    into this cell. Next, for all the points in such a grid cell, we
    find the closest point to $\hub(\SetA)$ in linear time.

    This takes $O( \cardin{\SetA} + \cardin{\SetB})$ time for the pair
    $\brc{X,Y}$. Overall, this takes time linear in the total weight
    of the given \SSPD.
\end{proof}

\section{Lower bound for \SSPD constructed using \BAR trees}
\apndlab{lower:bound}

Here, we demonstrate that a point might participate in a linear number
of pairs in the \SSPD if one uses the previous construction of Abam
\etal \cite{adfg-rftgs-09}.

\begin{theorem}
    There is a configuration of $n$ points in the plane, such that a
    point appearing in $n-1$ pairs of the \SSPD constructed using the
    algorithm of Abam \etal \cite{adfg-rftgs-09}.
\end{theorem}

\begin{proof}
    For the sake of simplicity of exposition, we assume that $n$ is a
    power of $2$.

    \begin{figure}[h]
        \centering\includegraphics{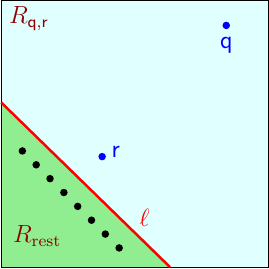}
        \caption{}
        \figlab{b_t}
    \end{figure}

    Consider the configuration illustrated in \figref{b_t}, where $n-2$ points are on a line making a $135$-degree angle with the $x$-axis, and the two other points $\pntA$ and $\pntB$ are far from them.  Let $\PntSet$ denote this set of points.

    The \BAR-tree construction algorithm \cite{dgk-bartc-01}, used on
    $\PntSet$, first separates the points $\pntA$ and $\pntB$ from the
    rest of the points by a $135$-degree splitting line $\ell$.  This
    splitting line produces two fat regions. One region, denoted by
    $\regionAB$, contains $\pntA$ and $\pntB$, and the other region,
    denoted by $\regionRest$, contains the rest of the points. The
    diameter of points inside $\regionAB$ is $\distX{\pntA}{\pntB}$,
    and it is large compared to the diameter of points inside
    $\regionRest$ or any subsequent subcell created inside
    $\regionRest$.

    Now, a node $v$ (and its corresponding region) of a \BAR tree is
    of \emphi{weight class} $i$, if $\cardin{\PntSet_v} \leq n/2^i$
    and $\cardin{ \PntSet_{\parentX{v}} } > n/2^i$.  As such, the
    region $\regionAB$ appears in all the weight classes for $i = 1,
    \cdots, \lg n - 1$, see \cite{adfg-rftgs-09}. The algorithm of
    Abam \etal \cite{adfg-rftgs-09} creates pairs only between nodes
    that belong to the same weight class (a node might belong to
    several weight classes).

    Observe that by placing the points of $\regionRest$ sufficiently
    close to the splitting line $\ell$, one can guarantee that any
    boundary of a subcell of $\regionRest$ in the \BAR-tree intersects
    $\ell$. Namely, no subcell of $\regionRest$ can be semi-separated
    from $\regionAB$, unless it contains a single point of $\PntSet$
    (and then its diameter is treated as being zero). Therefore,
    semi-separated pairs involving $\pntA$ and $\pntB$ that are
    produced in the weight class $i = 1, \cdots, \lg n - 1$ are of the
    form $(\{\pntA,\pntB\}, \{\pntC\})$ where $\pntC \in \PntSet
    \setminus \brc{\pntA, \pntB}$.  Moreover, sets involved in the
    produced semi-separated pairs in the weight class $\lg n$ are
    singletons. Therefore, both $\pntA$ and $\pntB$ appear in $n-2$
    pairs which all together semi-separate $\{\pntA,\pntB\}$ from
    $\PntSet \setminus \brc{\pntA, \pntB}$. Since the algorithm also
    generates the pair $\brc{ \brc{\pntA},\brc{\pntB}}$, these two
    points participate in $n-1$ pairs.
\end{proof}

\end{document}